\documentclass[11pt,a4paper]{article}
\usepackage[mathscr]{eucal}
\usepackage[latin1]{inputenc}
\usepackage[T1]{fontenc} 
\usepackage[english]{babel}
\usepackage{amscd,amssymb,amsfonts}
\usepackage{latexsym}
\usepackage{makeidx} 
\usepackage{ifthen}
\usepackage{graphicx}
\usepackage{amsmath}
\usepackage{color}
\usepackage{fancyhdr}
\usepackage{multirow}
\usepackage{eurosym}
\usepackage{textcomp} 
\usepackage{chngpage}
\usepackage{lscape}
\usepackage{layout}
\usepackage{indentfirst}
\usepackage{theorem}
\usepackage{url}

\usepackage[vcentering,pdftex,dvips]{geometry}

\numberwithin{equation}{section} 
\newtheorem{theorem}{Theorem}

\newtheorem{definition}[theorem]{Definition}

\newtheorem{proposition}[theorem]{Proposition}

\newenvironment{proof}[1][Proof]{\noindent\textbf{#1.} }{\ \rule{0.5em}{0.5em}}

\def\f{{\cal F}}
\newcommand{\field}[1]{\mathbb{#1}}
\newcommand{\bE}{\field{E}}
\newcommand{\bQ}{\field{Q}}

\title{The Multivariate Mixture Dynamics Model: \\Shifted dynamics and correlation skew}
\author{
Damiano Brigo\thanks{Dept. of Mathematics, Imperial College, London. damiano.brigo@imperial.ac.uk} \hspace{.5cm} Camilla Pisani\thanks{Dept. of Economics and Business Economics, Aarhus University, Denmark. The research leading to these results has received funding from the People Programme (Marie Curie Actions) of the European Union's Seventh Framework Programme FP7/2007-2013/ under REA grant agreement $n^{\circ}$ 289032. This paper  however reflects solely the Author's personal opinion and the Union is not liable for any use that may be made of the information contained therein. camilla.pisani@gmail.com.}  \hspace{.5cm} Francesco
Rapisarda\thanks{Bloomberg. This paper reflects solely the Author's personal opinion and does not represent the opinions of the author's  employers, present and past, in any way. frapisarda6@bloomberg.net} }

\date{\small{ First version:  December 15, 2015. This version: \today }}

\begin{document}
\maketitle


\begin{abstract}
The  Multi Variate Mixture Dynamics model is a tractable, dynamical, arbitrage-free multivariate model characterized by transparency on the dependence structure, since closed form formulae for terminal correlations, average correlations and copula function are available. It also allows for complete decorrelation between assets and instantaneous variances.  
Each single asset is modelled according to a lognormal mixture dynamics model, and this univariate version is widely used in the industry due to its flexibility and accuracy. The same property holds for the multivariate process of all assets, whose density is a mixture of multivariate basic densities.
This allows for consistency of single asset and index/portfolio smile.

In this paper, we generalize the MVMD model by introducing shifted dynamics and we propose a definition of implied correlation under this model. We investigate whether the model is able to consistently reproduce the implied volatility of FX cross rates once the single components are calibrated to univariate shifted lognormal mixture dynamics models. We consider in particular the case of the Chinese Renminbi FX rate, showing that the shifted MVMD model correctly recovers the CNY/EUR smile given the EUR/USD smile and the USD/CNY smile, thus highlighting that the model can also work as an arbitrage free volatility smile extrapolation tool for cross currencies that may not be liquid or fully observable. 

We compare the performance of the shifted MVMD model in terms of implied correlation with those of the shifted Simply Correlated Mixture Dynamics model where the dynamics of the single assets are connected naively by introducing correlation among their Brownian motions.   
Finally, we introduce a model with uncertain volatilities and correlation. The Markovian projection of this model is a generalization of the shifted MVMD model.

\bigskip

{\bf Key words:} MVMD model, Mixture of densities, Multivariate local volatility, Correlation Skew, Random Correlation, Calibration, Cross exchange rates, FX smile, Index volatility smile, renminbi-USD smile, renminbi-EUR smile, CNY-USD smile, CNY-EUR smile, SCMD model

\bigskip

{\bf AMS classification codes}: 60H10, 60J60, 62H20, 91B28, 91B70. 
\bigskip

{\bf JEL}: G13.
\end{abstract}

\section{Introduction to the Multivariate Mixture Dynamics}
The Multi Variate Mixture Dynamics model (MVMD) introduced by Brigo, Mercurio and Rapisarda \cite{MVMDWorkingPaper} and recently described in a deeper way in Brigo, Rapisarda and Sridi \cite{MVMD} is a tractable dynamical arbitrage-free model defined as the multidimensional version of the lognormal mixture dynamics model (LMD) in \cite{mixture1} and \cite{mixture2} (see also \cite{mixtureSartorelli}). The single-asset LMD model is a no-arbitrage model widely used among practitioners because of its practical advantages in calibration and pricing (analytical formulae for European options, explicit expression for the local volatility) and of its flexibility and accuracy. In fact, a variant of this model is presently used in the calibration of implied volatility surfaces for single stocks and equity indices in the Bloomberg terminal \cite{RefBloomberg}, and in the subsequent pricing of European, American and path-dependent options on single assets and baskets of assets. The main advantage of the MVMD over other multidimensional models, such as e.g., the Wishart model (\cite{fonseca07} and \cite{gurieroux07}) is in its tractability and flexibility which allows the MVMD to calibrate index volatility smiles consistently with the univariate assets smiles. In addition, a full description of its dependence structure (terminal correlations, average correlations, copula functions) is available. 

The MVMD model also enjoys some interesting properties of Markovian projection. First of all, the model can be seen as a Markovian projection of a model with uncertain volatilities denominated MUVM model. As a consequence, European option prices under the MVMD model can more easily be computed under the MUVM model instead. However, the MVMD model remains superior in terms of smoothness and dynamics. 
Secondly, the Geometric average basket under the MVMD model can be projected into a univariate lognormal mixture dynamics model. Consequently, European option prices on the basket can be easily computed through the Black and Scholes formula.  

Finally, under the MVMD model, the terminal correlation between assets and squared volatilities is zero. This mitigates the common drawback of local volatility models of having perfect instantaneous correlation between assets and squared volatilities.

In this paper we generalize the MVMD model, including shifts to the dynamics of the single assets, and we study the correlation skew under this framework.

Before going into details, we recapitulate the definition of the MVMD model (in the non-shifted case), starting with the univariate LMD model and then generalizing to the multidimensional case. 

\subsection{The volatility smile mixture dynamics model for single assets}

Given a maturity $T>0$, we denote by $P(0,T)$ the price at time $0$ of
the zero-coupon bond maturing at $T$, and by $(\Omega,\f,\mathbb{P})$ a probability space with a filtration $({\f}_t)_{t\in[0,T]}$ which is $\mathbb{P}$-complete and satisfying to the usual conditions. We assume the existence of a measure $\mathbb{Q}$ equivalent to $\mathbb{P}$, called the risk--neutral or pricing measure, ensuring arbitrage freedom in the classical setup, for example, of Harrison,  Kreps and Pliska \cite{harrison&kreps,harrison&pliska}.
In this framework we consider $N$ purely \textit{instrumental} diffusion processes $Y^i(t)$ with dynamics
\begin{equation}\label{lognormal_EDS}
d Y^i(t) = \mu  Y^i(t) dt + v^i(t, Y^i(t))Y^i(t) dW(t)
\end{equation}
and a deterministic initial value $Y^i(0),$ marginal densities $p_t^{i}$ and diffusion coefficient $v_i$. We define $S_t$ as the solution of 
\begin{equation}
dS(t) = \mu S(t) dt + s(t,S(t)) S(t) dW(t) \label{localVol}
\end{equation}
where $s$ is a local volatility function, namely a deterministic function of $t$ and $S$ only, and it is computed so that the marginal density $p_t$ of $S(t)$ is 
a linear convex combination of the densities $p_t^{i}$ 
\cite{mixture1,mixture2,mixtureFX}:
\begin{equation}
p_t = \sum_i \lambda^i p_t^{i} \hspace{.2cm} \mbox{with}
\hspace{.2cm} \lambda^i \ge 0, \forall i \hspace{.2cm} \mbox{and}
\hspace{.2cm} \sum_i \lambda^i = 1. \label{mixtureOne}
\end{equation}

In what follows we restrict ourselves to the case
\begin{equation}\label{sgherlo}
\left\{
  \begin{array}{l l}
Y^i(0)&=S(0),\\
v_i(t,x) &= \sigma^i(t), \\ 
V^i(t) &= \sqrt{\int_0^t \sigma^i(s)^2 ds} \\
p_t^{i}(x) &= \frac{1}{\sqrt{2 \pi} x V^i(t)} \exp\left[ -\frac{1}{2
V_i^2(t)} \left( \ln\left(\frac{x}{S(0)}\right) - \mu t +
\frac{1}{2} V^i(t)^2 \right)^2 \right] = \ell^i_t(x)
 \end{array} \right.
\end{equation}
with $\sigma^i$ deterministic. 
The parameter $\mu$ is completely specified by $\mathbb{Q}$. If the asset is a stock paying a continuous dividend yield $q$ and $r$ is the time $T$ constant risk-free rate, then $\mu=r-q$. If the asset is an exchange rate and $r_d$ and $r_f$ are the (deterministic) domestic and foreign rates at time $T$, respectively, then $\mu=r_d-r_f$. If the asset is a forward price, then $\mu=0$.

Brigo and Mercurio \cite{mixture2} proved that defining 
\begin{equation}\label{nu}
s(t,x)  =  \left(\frac{\sum_{k=1}^N
\lambda^{k} \sigma^k(t)^2
\ell^k_t(x)}{\sum_{k=1}^N \lambda^{k}
\ell^k_t(x)}\right)^{1/2}
\end{equation}
and assuming a few additional nonstringent assumptions on the $\sigma^i$, the corresponding dynamics for $S_t$ admits a unique strong solution. 

\begin{theorem}\label{th:LMDE} {\textbf{Existence and uniqueness of solutions for the LMD model}}. Assume that all the real functions $\sigma^i(t)$, defined on the real numbers $t \ge 0$, are once continuously differentiable and bounded from above and below by two positive real constants. Assume also that in a small initial time interval $t \in [0, \epsilon]$, $\epsilon >0$, the functions $\sigma^i(t)$ have an identical constant value $\sigma_0$. Then the Lognormal Mixture Dynamics model (LMD) defined by
\begin{equation}\label{eq:dcmix}  
d S_t = \mu S_t dt + s(t,S_t) S_t dW_t, \ \ S_0, \ \  s(t,x)  =  \left(\frac{\sum_{k=1}^N
\lambda^{k} \sigma^k(t)^2
\ell^k_t(x)}{\sum_{k=1}^N \lambda^{k}
\ell^k_t(x)}\right)^{1/2},
\end{equation}
admits a unique strong solution and the forward Kolmogorov equation (Fokker Planck equation) for its density admits a unique solution satisfying \eqref{mixtureOne}, which is a mixture of lognormal densities.
\end{theorem}

An important consequence of the above construction is that European option prices on $S$ can be written as linear combinations of Black-Scholes prices with weights $\lambda^i$. The same combination holds for the Greeks at time $0$.

\subsection{Combining mixture dynamics on several assets: SCMD}

Consider now $n$ different asset prices $S_1 \dots S_n$ each calibrated to an LMD model, as in equation (\ref{eq:dcmix}), and denote by $\lambda_i^{k}$, $\sigma_i^{k}$ the parameters relative to the $k$-th instrumental process of the asset $i$. There are two possible ways in order to connect the dynamics of the single assets into a multivariate model. The first more immediate way consists in introducing a non-zero quadratic covariation between the Brownian motions driving the LMD models of equation (\ref{eq:dcmix}) for $S_1 \dots S_n$ leading to the so-called SCMD model. 
 
\begin{definition} {\bf SCMD Model}.
We define the Simply Correlated multivariate Mixture Dynamics (SCMD) model
for $\underline{S}= [S_1,\ldots,S_n]$  as a vector of univariate LMD models, each satisfying Theorem
\ref{th:LMDE} with diffusion coefficients $s_1,\ldots,s_n$ given by equation  \eqref{eq:dcmix} and densities $\ell_1,\ldots,\ell_n$ applied to each asset, and connected simply through quadratic covariation $\rho_{ij}$ between the Brownian motions driving assets $i$ and $j$.
This is equivalent to the following $n$-dimensional diffusion process where we keep the $W$'s independent and where we embed the Brownian covariation into the diffusion matrix $\bar{C}$, whose $i$-th row we denote by $\bar{C}_i$:
\begin{equation}\label{edsSCMD1}
d\underline{S}(t) = diag(\underline{\mu}) \underline{S}(t)dt +
diag(\underline{S}(t))
\bar{{C}}(t,\underline{S}(t))d\underline{W}(t), \ \ \  \bar{a}_{i,j}(t,\underline{S}) := \bar{{C}}_{i} \bar{{C}}_{j}^T
\end{equation}
\begin{equation}\label{aSCMD}
\bar{a}_{i,j}(t,\underline{S}) = s_i(t,S_i) s_j(t,S_j)  \rho_{ij}= \left(\frac{\sum_{k=1}^N
\lambda_i^{k} \sigma_i^k(t)^2
\ell_{i,t}^k(S_i)}{\sum_{k=1}^N \lambda_{i}^{k}
\ell_{i,t}^k(S_i)} \ \  \frac{\sum_{k=1}^N \lambda_j^k
\sigma_{j}^k(t)^2 \ell_{j,t}^k(S_j)}{\sum_{k=1}^N
\lambda_{j}^k \ell_{j,t}^k(S_j)}\right)^{1/2} \rho_{ij}
\end{equation}
where $T$ represents the transposition operator.
\end{definition}

\noindent{\bf Assumption.} We assume $\rho=(\rho_{ij})_{i,j}$ to be positive definite.

\qquad

It is evident from the previous construction that the SCMD is consistent with both the dynamics of the single assets $S_i$ and the instantaneous correlation matrix $\rho$. Moreover, we can easily simulate a path of $S$ by exogenously computing $\rho$ for example from historical data, assuming it constant over time and applying a naive Euler scheme.
However an explicit expression for the density of $\underline{S}= [S_1,\ldots,S_n]$ under the SCMD dynamics is not available. As a consequence, if we aim at computing prices of options whose payoff depends on the value at time $T$ only we still need to simulate entire paths of $\underline{S}$ over the interval $[0,T]$, which can be quite time consuming. 

\subsection{Lifting the mixture dynamics to asset vectors: MVMD}

A different approach, still consistent with the single assets' dynamics, lies in merging the dynamics of the single assets in such a way that the mixture property is lifted to the multivariate density and the corresponding model gains some further tractability property with respect to the SCMD model. This can be achieved by mixing in all possible ways the densities of the instrumental processes of each individual asset and by imposing the correlation structure $\rho$ at the level of the single instrumental processes, rather than of the assets as we did for the SCMD model. This has important consequences on the actual structure of the correlation, see \cite{MVMDWorkingPaper}. Below we summarize the construction leading to the MVMD model, while referring to Brigo et al. \cite{MVMD} for further details.

Assume we have calibrated an LMD model for each $S_i(t)$:
if $p_{S_i(t)}$ is the density of $S_i$, we write
\begin{equation}
p_{S_i(t)}(x) = \sum_{k=1}^{N_i} \lambda_{i}^{k} \ell_{i,t}^k(x),
\hspace{.2cm} \mbox{with} \hspace{.2cm} \lambda_{i}^k \ge 0, \forall
k \hspace{.2cm} \mbox{and} \hspace{.2cm} \sum_k \lambda_{i}^k = 1,
\label{individualDensity}
\end{equation}
where $(\ell_{i,t}^k)_{k}$ are the densities of $(Y_i^{k})_k$, instrumental processes for $S_i$
evolving lognormally according to the stochastic differential
equation:
\begin{equation}\label{edsYi_k}
dY_i^{k}(t) = \mu_i Y_i^{k}(t) dt + \sigma_i^{k}(t) Y_i^{k}(t)
dZ_i(t),\ \ \ d \langle Z_i, Z_j\rangle_t = \rho_{ij} dt,\ \ \ Y_i^k(0)=S_i(0).
\end{equation}
For notational simplicity we assume the number of base
densities $N_i$ to be the same, $N$, for all assets. The exogenous
correlation structure $\rho_{ij}$ is given by the symmetric,
positive--definite matrix $\rho$.

Denote by ${\underline{S}}(t)
= \left[S_1(t),\cdots,S_n(t)\right]^T$ the vector of asset prices with
\begin{equation}\label{edsMVMD}
d\underline{S}(t) = diag(\underline{\mu}) \underline{S}(t)dt +
diag(\underline{S}(t))
{A}(t,\underline{S}(t))d\underline{W}(t).
\end{equation}
As we did for the one dimensional case, we look for a matrix $A$ 
such that
\begin{equation}
p_{\underline{S}(t)}(\underline{x})=
\sum_{k_1,k_2,\cdots k_n=1}^N
\lambda_1^{k_1} \cdots \lambda_n^{k_n}
\ell_{1,\ldots,n; t}^{k_1,\ldots,k_n}(\underline{x}), \ \ \ \ \ell_{1,\ldots,n;t}^{k_1,\ldots,k_n}(\underline{x}) := p_{\left[Y_1^{k_1}(t),...,Y_n^{k_n}(t)\right]^T}(\underline{x}),
\label{mixture_k}
\end{equation}
or more explicitly
\begin{eqnarray}
\ell_{1,\ldots,n; t}^{k_1,\ldots,k_n}(\underline{x}) =
\frac{1}{(2 \pi)^{\frac{n}{2}} \sqrt{\det \Xi^{(k_1 \cdots k_n)}(t)}
\Pi_{i=1}^n x_i} & & \exp\left[ -\frac{ \tilde{x}^{(k_1 \cdots k_n)
T} \Xi^{(k_1 \cdots k_n)}(t)^{-1} \tilde{x}^{(k_1 \cdots k_n)} }{2}
\right] \nonumber\label{multivariateGaussian}
\end{eqnarray}
where $\Xi^{(k_1 \cdots k_n)}(t)$ is the integrated covariance
matrix whose $(i,j)$ element is
\begin{equation}
\Xi_{ij}^{(k_1 \cdots k_n)}(t) = \int_0^t \sigma_i^{k_i}(s)
\sigma_j^{k_j}(s) \rho_{ij} ds \label{covarianceMatrix}
\end{equation}
\begin{equation}\label{xtild_dimn}
\tilde{x}_i^{(k_1 \cdots k_n)}=\ln
 x_i - \ln x_i(0) - \mu_i t + \int_0^t \frac{\sigma_i^{k_i^2}(s)}{2} ds.
\end{equation}

Computations show that if a solution exists, this must satisfy the definition below.

\begin{definition}\label{DefMVMD} {\bf MVMD Model}. 
The (Lognormal) Multi Variate Mixture Dynamics (MVMD) model is given by 
\begin{eqnarray}\label{edsMVMDwithindBM1}
d\underline{S}(t) &=&  diag(\underline{\mu})\ \underline{S}(t)\ dt +
diag(\underline{S}(t))\
{C}(t,\underline{S}(t)) B \ d\underline{W}(t), \\  \nonumber
 C_i(t,\underline{x}) &:=&\frac{\sum_{k_1,...,k_n=1}^{N} \lambda_1^{k_1}...\lambda_n^{k_n}\ \sigma_i^{k_i}(t)  \
\ell_{1,\ldots,n;t}^{k_1,\ldots,k_n}(\underline{x})}{\sum_{k_1,...,k_n=1}^{N}
\lambda_1^{k_1}...\lambda_n^{k_n}\
\ell_{1,\ldots,n;t}^{k_1,\ldots,k_n}(\underline{x})},
\end{eqnarray}
$\ell_{1,\ldots,n;t}^{k_1,\ldots,k_n}(\underline{x}) := p_{\left[Y_1^{k_1}(t),...,Y_n^{k_n}(t)\right]^T}(\underline{x})$ and defining $B$ such that $\rho=B B^T$, $a = C B (C B)^T$, 
\begin{equation}\label{C}
a_{i,j}(t,\underline{x}) 
= \frac{\sum_{k_1,...,k_n=1}^{N} \lambda_1^{k_1}...\lambda_n^{k_n}\
{V}^{k_1,...,k_n}(t)\
\ell_{1,\ldots,n;t}^{k_1,\ldots,k_n}(\underline{x})}{\sum_{k_1,...,k_n=1}^{N}
\lambda_1^{k_1}...\lambda_n^{k_n}\
\ell_{1,\ldots,n;t}^{k_1,\ldots,k_n}(\underline{x})}
\end{equation}
where
\begin{equation}\label{V}
{V}^{k_1,...,k_n}(t) = 
 \left[\sigma_i^{k_i}(t)\ \rho_{i,j}\
\sigma_j^{k_j}(t)\right]_{i,j = 1,...,n}.
\end{equation}
\end{definition}
From the previous definitions it is evident that the dynamics of the single assets $S_i$ in the SCMD model are Markovian. On the other hand, under the MVMD model, while the dynamics of the whole vector $S$ is Markovian, those of the single assets are not. This leads to more realistic dynamics.

Under mild assumptions, existence and uniqueness  of a solution can be proved through the following Theorem.
\begin{theorem}
Assume that the volatilities $\sigma_i^{k_i}(t)$ for all $i$ are once continuously differentiable, uniformly bounded from below and above by two positive real numbers  $\tilde{\sigma}$ and $\hat{\sigma}$ respectively, and that they take a common constant value $\sigma_0$ for $t \in [0,\epsilon]$ for a small positive real number $\epsilon$,  namely
\begin{eqnarray*}
&& \tilde{\sigma}=\inf_{t \ge 0} \left( \   \   \    \min_{i=1 \cdots n, k_i=1, \cdots
N} \ \ (
\sigma_i^{k_i}(t) ) \right), \\ 
&& \hat{\sigma}=\sup_{t \ge 0} \left( \ \ \  \max_{i=1 \cdots n, k_i=1 \cdots N} \ \ (
\sigma_i^{k_i}(t) ) \right) \\
&& \sigma_i^{k_i}(t)  = \sigma_0 > 0 \ \ \mbox{for all} \ \ t \in [0, \epsilon].
\end{eqnarray*}
Assume also the matrix $\rho$ to be positive definite. Then the MVMD n-dimensional stochastic differential equation (\ref{edsMVMDwithindBM1})
admits a unique strong solution. The diffusion matrix $a(t,\underline{x})$ in (\ref{C}) is positive definite for all $t$ and $x$.
\end{theorem}

\section{Introducing a shift in MVMD}
When modelling a one dimensional asset price through an LMD model, implied volatilities with minimum exactly at a strike equal to the forward asset price are the only possible. In order to gain greater flexibility and therefore move the smile minimum point from the ATM forward we can shift the overall density by a deterministic
function of time, carefully chosen in order to preserve risk--neutrality and therefore guarantee no--arbitrage. 
This is the
so--called {\em shifted lognormal mixture dynamics model} \cite{mixtureFX}. Under this model the new asset-price process $S$ is defined as
\begin{equation}\label{eq:shiftedLMD}
 S_t= \beta e^{\mu t} +X_t
\end{equation}
with $\beta$ real constant and $X_t$ satisfying (\ref{eq:dcmix}). Under the assumption $K- \beta e^{\mu T}>0$ the price at time $0$ of a European call option with strike $K$ and maturity $T$ can be written as
\begin{equation}
P(0,T) \mathbb{E}^T\{ (S_T-K)^{+}\}=P(0,T) \mathbb{E}^T\{ (X_T-[K- \beta e^{\mu T}])^{+}\} 
\end{equation}
and thus as a combination of Black and Scholes prices with strike $K- \beta e^{\mu T}$. The model therefore preserves the same level of tractability as in the non shifted case with the advantage of gaining more flexibility.

Once each asset is calibrated to a shifted LMD model, we have two possibilities for reconstructing the dynamics of the multidimensional process. The first possibility is to reconnect the single assets by introducing a non-zero quadratic covariation between the Brownian motions (as we did for the SCMD model), leading to what we call the \textit{shifted SCMD model}. 
The second possibility going on the same lines as the approach leading to the MVMD model, lies in applying the same shift $\beta_i e^{\mu_i t}$ to each instrumental process $Y_i^k$ of each asset $X_i$ 
$$S_i^k(t)=Y_i^k(t)+\beta_i e^{\mu_i t}$$
where $Y_i^k$ satisfies the dynamics in (\ref{edsYi_k}) (this is equivalent to applying the shift $\beta_i e^{\mu_i t}$ directly to the $i$-th asset)
and then mix the corresponding densities $p_{S_i^k(t)}(x)$ in all possible ways. Computations similar to those for the non-shifted case show that if a solution exists, it must satisfy the definition below (details on the computations are shown in the Appendix).

\begin{definition}{\bf Shifted MVMD Model}\label{def:ShiftedMVMD}.
The shifted Multi Variate Mixture Dynamics model is given by
\begin{eqnarray}\label{edsShiftedMVMD}
d \underline{S}(t)&=&diag(\underline{\mu})\underline{S}(t) dt + diag(\underline{S}(t)) \widetilde  C(t, \underline{S}(t))B d \underline{W}(t), \\  \nonumber
\widetilde  C_i(t,\underline{x}) &:=&\frac{\sum_{k_1,...,k_n=1}^{N} \lambda_1^{k_1}...\lambda_n^{k_n}\ \sigma_i^{k_i}(t) (x_i-\beta_ie^{\mu_i t}) \
\tilde \ell_{1,\ldots,n;t}^{k_1,\ldots,k_n}(\underline{x})}{x_i \sum_{k_1,...,k_n=1}^{N}
\lambda_1^{k_1}...\lambda_n^{k_n}\
\tilde \ell_{1,\ldots,n;t}^{k_1,\ldots,k_n}(\underline{x})},
\end{eqnarray}
\begin{equation}\label{ell tilde}
\tilde{\ell}_{1, \dots n;t}^{k_1, \dots , k_n} (\underline{x}) = p_{[S_1^{k_1}(t), \dots , S_n^{k_n}(t)]^T}( \underline{x})=\ell_{1,\ldots,n;t}^{k_1,\ldots,k_n}(\underline{x}-\underline{\beta} e^{\underline{\mu} t})
\end{equation}
and defining $B$ such that $\rho=B B^T$, $\tilde a=\widetilde  C B (\widetilde  C B)^T$,
\begin{equation}\label{CShiftedMVMD}
\tilde a_{ij}(t,\underline{x})=\frac{ \sum_{k_1,k_2,\dots k_n = 1}^N \lambda_1^{k_1} \cdots \lambda_n^{k_n}  \
{V}^{k_1,...,k_n}(t) (x_i-\beta_i e^{\mu_i t})(x_j-\beta_j e^{\mu_j t}) \tilde{\ell}_{1, \dots n;t}^{k_1, \dots , k_n} (\underline{x})}{x_i x_j \sum_{k_1,k_2,\dots k_n = 1}^N \lambda_1^{k_1} \cdots \lambda_n^{k_n} \tilde{\ell}_{1, \dots n;t}^{k_1, \dots , k_n} (\underline{x})} 
\end{equation}
with ${V}^{k_1,...,k_n}$ as in (\ref{V}).
\end{definition}

We now have all the instruments to introduce the correlation skew and study its behaviour under shifted SCMD and shifted MVMD dynamics.

\section{The correlation skew}
The aim of this section is to introduce a definition of \textit{correlation skew} and to study its behaviour under shifted MVMD dynamics, in comparison with the correlation skew under shifted SCMD dynamics.
It is observed in practice under normal market conditions that assets are relatively weakly correlated with each other. However during periods of market stress stronger correlations are observed. This fact suggests that a single correlation parameter for all options quoted on a basket of assets, or an index, say, may not be sufficient to reproduce all option prices on the basket/index for a given expiry. In fact, this is what is observed empirically when inferring a multidimensional dynamics from a set of single--asset dynamics. Among others, this has been shown in Bakshi et al. \cite{bakshi} for options on the S\&P 100 index and in Langnau \cite{langnau} for options on the Euro Stoxx 50 index and on the DAX index.

When computing the implied volatility, European call prices (or equivalently put prices) are considered and the reference model is the benchmark Black \& Scholes \cite{black_scholes} model. It seems then natural to consider as multidimensional benchmark a model where the single assets follow geometric Brownian motions and constant correlation among the single Brownian shocks is introduced.
However, when moving from the one-dimensional to the multidimensional framework a bigger variety of possible option instruments to use in order to compare prices under the reference model and the model under analysis appears, the particular choice depending on the specific product we are interested in. 
Austing \cite{Austing} recently provided a discussion on some of the most popular multi-assets products suggesting the use of composite options as benchmark on which defining the implied correlation.
In this paper we adopt a different approach based on the comparison with options on $S_1(t)$, $S_2(t)$ with payoff 
\begin{equation}\label{eq:payoffBasket}
\left(S_1(T)S_2(T)-K\right)^{+}.
\end{equation}

Assume that the pair $(S_1,S_2)$ follows a bi-dimensional Black and Scholes model, in other words $S_1$ and $S_2$ follow two geometric Brownian motions with correlation $\rho$ and consider the payoff in equation (\ref{eq:payoffBasket}). Given the Black and Scholes implied volatilities for $S_1$ and $S_2$, the value $\rho_{impl}$ such that prices under the bi-dimensional Black and Scholes model are the same as market prices
\begin{equation*}
\mbox{MKT\_Prices}(S_1(0),S_2(0),K,T)=\mbox{BS\_Prices}(S_1(0),S_2(0),K,T, \rho_{impl}(K,T))
\end{equation*}
is called \textit{implied correlation}. If we try to match option prices for a given maturity $T$ and two different strikes $K_1$, $K_2$, we will observe two different values of the implied correlation. This is contrary to the hypothesis of constant correlation in the bi-dimensional Black and Scholes model. 

The curve $K \to \rho_{impl}(K,T)$ is called correlation skew. Thus, the correlation skew can be considered as a descriptive tool/metric similar to the volatility smile in the one-dimensional case, with the difference that it primarily describes implied dependence instead of volatility.

\subsection{Explaining the skew in MVMD with the single parameter $\rho$ via MUVM} \label{sec:CorrSkewMVMD}
The aim of this section is to introduce a definition of implied correlation under shifted MVMD dynamics, when using options with payoff as in equation (\ref{eq:payoffBasket}). This leads to a straightforward application in the foreign exchange market within the study of triangular relationships. Imagine, for example, that $S_1$ and $S_2$ represent the exchange rates USD/EUR and EUR/JPY, respectively. The cross asset $S_3=S_1 S_2$ then represents the USD/JPY exchange rate, and the corresponding payoff in equation (\ref{eq:payoffBasket}) is the payoff of a call option on the USD/JPY FX rate. In the following, we will investigate whether the shifted MVMD model is able to consistently reproduce the implied volatility of $S_3$, once the single components $S_1$, $S_2$ are calibrated to univariate shifted LMD models. Consistency properties of this kind are important, for example, in order to reconstruct the time series of less liquid cross currency pairs from more liquid ones.

Before proceeding we make a remark on the interpretation of $\rho$. Keeping in mind the definition of instantaneous local correlation in a bivariate diffusion model
\begin{equation*}
\rho_L(t) := \frac{d \langle S_1, S_2 \rangle_t}{\sqrt{d \langle S_1, S_1 \rangle_t \ d \langle S_2, S_2 \rangle_t } } 
\end{equation*}
and making use of Schwartz's inequality, we obtain that the absolute value of the local correlation under the shifted MVMD model is smaller than the value under the shifted SCMD model. The result is contained in the Proposition below. 

\begin{proposition}[Local correlation in shifted MVMD and shifted SCMD]\label{Prop:LocalCorrelation}
The instantaneous local correlation under the shifted SCMD model is $\rho$, whereas for the shifted MVMD model we have
\begin{align*}
\rho_L(t) &=\frac{
\rho \sum_{k,k'=1}^{N} {\lambda_1}^k {\lambda_2}^{k'} \sigma_1^{(k)}
\sigma_2^{(k')}  \tilde \ell_t^{(kk')}(x_1,x_2)
}
{
\sqrt{
\left(
\sum_{k,k'=1}^{N} {\lambda_1}^k {\lambda_2}^{k'}
\sigma_1^{(k)2} \tilde \ell_t^{(kk')}(x_1,x_2)
\right)
\left(
\sum_{k,k'=1}^{N} {\lambda_1}^k {\lambda_2}^{k'}
\sigma_2^{(k')2} \tilde \ell_t^{(kk')}(x_1,x_2)
\right)
}
}, \cr
\vert \rho_L(t) \vert &\le \rho
\end{align*}
where $\tilde \ell_t^{(kk')}(x_1,x_2)$ is defined as in equation (\ref{ell tilde}).
\end{proposition}
We see that $\rho$ enters the formula for  the instantaneous local correlation $\rho_L$ in the MVMD model, even though the latter is more complex than the constant value $\rho$.  
Our aim is to find a value of $\rho$ matching prices of options with payoff as in equation (\ref{eq:payoffBasket}) under shifted MVMD dynamics with market prices.

In order to do that we will make use of a model with uncertain parameters of which the shifted MVMD is a Markovian projection. Indeed, as shown in Brigo et al. \cite{MVMD} the MVMD model as in Definition \ref{DefMVMD} (without shift) is a Markovian projection of the model defined below

\begin{equation}\label{edsuncertain_volatily_model}
d\xi_i(t) = \mu_i\ \xi_i(t) dt + \sigma^{I_i}_i(t)\ \xi_i(t) dZ_i(t),\ \ i
= 1,...,n,
\end{equation}

where each $Z_i$ is a
standard one dimensional Brownian motion with $d \left \langle Z_i,
Z_j \right \rangle_t = \rho_{i,j} dt$, $\mu_i$ are
constants, $\sigma^I:= [\sigma^{I_1}_1,\ldots,\sigma^{I_n}_n]^T$ is a random vector independent of $Z$ and representing uncertain volatilities with 
$I_1,\ldots,I_n$ mutually independent.
More specifically,
each $\sigma^{I_i}_i$ takes values in a set of $N$ deterministic functions
$\sigma_i^k$ with probability $\lambda_i^k$. Thus, for all times in $(\varepsilon,+\infty)$ with small $\varepsilon$ we have
\begin{equation*}
(t \longmapsto \sigma^{I_i}_i(t)) = \left\{\begin{array}{l} (t \longmapsto
\sigma_i^1(t))\ \  \mbox{with} \ \mathbb{Q} \ \ \mbox{probability}\ \  \lambda_i^1 \\
(t \longmapsto
\sigma_i^2(t))\ \   \mbox{with} \ \mathbb{Q} \ \ \mbox{probability}\ \  \lambda_i^2 \\
\vdots \\
(t \longmapsto \sigma_i^N(t))\ \  \mbox{with} \ \mathbb{Q} \ \ \mbox{probability}\ \
\lambda_i^N
\end{array} \right.
\end{equation*}

Now it is straightforward to show that if we add a shift to each component as follows
\begin{equation}\label{uncertain_volatily_model-shift}
\tilde \xi_i(t)=\xi_i(t)+\beta_i e^{\mu_i t}
\end{equation}
we obtain a model having the shifted MVMD model (\ref{edsShiftedMVMD})-(\ref{CShiftedMVMD}) as Markovian projection. This can be easily shown by Gy\"{o}ngy's lemma \cite{gyongi}.

\begin{theorem}\label{Theo:MarkovianProjectionShiftedModels}
The shifted MVMD model is a Markovian projection of the shifted MUVM model.
\end{theorem}

\begin{proof}
A straightforward application of Ito's lemma shows that $\tilde \xi(t)$ satisfies the system of SDEs below
\begin{equation}\label{edsscenariomodel-shift} d\underline{\tilde \xi}(t) = diag(\underline{\mu})\ \underline{\tilde \xi}(t)\ dt + diag (\underline{\tilde \xi}(t)-\underline{\alpha}(t))\ {A^I(t)}\
d\underline{W}(t)
\end{equation}
where $diag(\underline{\alpha}(t))$ is a deterministic matrix whose $i$-th diagonal element is the shift $\beta_i e^{\mu_i t}$ and  ${A^I(t)}$ is the Cholesky decomposition of the covariance matrix
$\Sigma^I_{i,j}(t) := \sigma^{I_i}_i(t) \sigma^{I_j}_j(t)\ \rho_{ij}.$

Define $\tilde v(t,\underline{\xi}(t)) = diag(\underline{\tilde \xi}(t)-\underline{\alpha}(t)) {A^I(t)}$. In order to show that the MVMD model is a Markovian projection of the MUVM model, we need to show that 

\begin{equation}\label{eq:Gyongy}
  \mathbb{E}[\tilde v \tilde v^{T}|\underline{\tilde \xi}(t) = \underline{\tilde x}]  = \tilde \sigma \
\tilde \sigma^T(t,\underline{x}).
\end{equation}
where $\tilde \sigma(t,\underline{x})=diag(\underline{x}) \widetilde C(t, \underline{x})B$ and $ \widetilde C$ is defined as in (\ref{edsShiftedMVMD}). 

Observing that 
\begin{eqnarray*}
& &\mathbb{E}[\tilde v \tilde v^{T}|\underline{\xi}(t) \in d \underline{x}] =
\frac{\mathbb{E}[diag(\underline{\tilde \xi}(t)- \underline{\alpha}(t))\ \Sigma\
diag(\underline{\tilde \xi}(t)- \underline{\alpha}(t))\ 1_{\{\underline{\tilde \xi}(t) \in
d\underline{x}\}}]}{\mathbb{E}[1_{\{\underline{\tilde \xi}(t) \in
 d\underline{x}\}}]} =\\
& & \frac{diag(\underline{x}-\underline{\alpha}(t)) \sum_{k_1,...,k_n = 1}^{N}
\lambda_1^{k_1}...\lambda_n^{k_n}\ {V}^{k_1,...,k_n}(t)
\tilde \ell_{1,\ldots,n;t}^{k_1,\ldots,k_n} (\underline{x})\ 
diag(\underline{x}-\underline{\alpha}(t)) \ d\underline{x}}{\sum_{k_1,...,k_n = 1}^{N} \lambda_1^{k_1}...\lambda_n^{k_n}\
\tilde \ell_{1,\ldots,n;t}^{k_1,\ldots,k_n} (\underline{x}) \ d\underline{x} }
\end{eqnarray*}
and performing simple matrix manipulations, equation (\ref{eq:Gyongy}) is easily obtained.
\end{proof}

Since we will infer the value of $\rho$ from prices of options with payoff as in (\ref{eq:payoffBasket}) depending on the value of $(S_1,S_2)$ at time $T$ only, we can make computations under the shifted MUVM rather then the shifted MVMD, as these two models have the same one-dimensional (in time) distributions. Computations under the shifted MUVM model are easier to do (with respect to the shifted MVMD case) since conditioning on $\{I_i=j\}$, $\xi_i$ follows a shifted geometric Brownian motion with volatility $\sigma_i^j$. 

In particular we will focus on the bidimensional specification in which case the shifted MUVM reduces to

\begin{equation}\label{eq:ShiftedMUVM}
\begin{array}{l l}
d S_1(t) = \mu_1\ S_1(t) dt + \sigma^{I_1}_1(t)\ (S_1(t)-\beta_1 e^{\mu_1 t}) dW_1(t)\\
d S_2(t) = \mu_2\ S_2(t) dt + \sigma^{I_2}_2(t)\ (S_2(t)-\beta_2 e^{\mu_2 t}) dW_2(t)
\end{array}
\end{equation}
where the Brownian motions $W_1$, $W_2$ have correlation $\rho$. 

Once we have calibrated $S_1$ and $S_2$ independently, each to a univariate shifted LMD model, we notice that the only parameter missing when computing prices of options having payoff as in (\ref{eq:payoffBasket}) is $\rho$.

\begin{definition}\label{def:ImpliedCorrelation}
We define the \textit{implied correlation parameter} in the shifted MVMD model as the value $\rho$ minimizing the squared percentage difference between implied volatilities from options with payoff (\ref{eq:payoffBasket}) under the  shifted MVMD model and market implied volatilities.
\end{definition}

\subsection{The correlation skew in SCMD via $\rho$} \label{sec:CorrSkewSCMD}
Now, assume to model the joint dynamics of $(S_1,S_2)$ as a shifted SCMD model instead. In this case
\begin{equation}\label{eq:sdeSCMD}
\begin{array}{l l}
d S_1(t) &= \mu S_1(t) dt + \nu_1(t,S_1(t)-\beta_1 e^{\mu t}) (S_1(t)-\beta_1 e^{\mu t}) dW_1(t),  \\
d S_2(t) &= \mu S_2(t) dt + \nu_2(t,S_2(t)-\beta_2 e^{\mu t}) (S_2(t)-\beta_2 e^{\mu t}) dW_2(t)
\end{array} 
\end{equation}
with
\begin{equation}\label{eq:sdeSCMD_nu}
\begin{array}{l l}
\nu_1(t,x) &=  \left(\frac{\sum_{k=1}^N
\lambda_1^{k} \sigma_1^k(t)^2
\ell^k_t(x)}{\sum_{k=1}^N \lambda_1^{k}
\ell^k_t(x)}\right)^{1/2},\\
\nu_2(t,x) &=  \left(\frac{\sum_{k=1}^N
\lambda_2^{k} \sigma_2^k(t)^2
\ell^k_t(x)}{\sum_{k=1}^N \lambda_2^{k}
\ell^k_t(x)}\right)^{1/2}
\end{array} 
\end{equation}
where the Brownian motions $W_1$, $W_2$ have correlation $\rho$. 
In this case the parameter $\rho$ really represents the true value of the instantaneous local correlation, as opposed to the MVMD case. We still define the \textit{implied correlation} as the value $\rho$ minimizing the squared percentage difference between implied volatilities from options with payoff (\ref{eq:payoffBasket}) under the shifted SCMD model and market implied volatilities.

\subsection{Pricing under the shifted MUVM}

Now, we consider computing the price of options such as (\ref{eq:payoffBasket}), namely options on cross FX rates, under the shifted model. In general one has a loss of tractability with respect to the non-shifted case. However, one can still express the price via a semi-analytic formula involving double integration:

\begin{align}\label{eq:PriceBaske}
\small
&e^{-r T} \bE [(B-K)_{+}]= \nonumber\\
&e^{-r T} \sum_{i,j=1}^N \lambda_1^i \lambda_2^j  \int_{K}^{\infty} dB(B-K) \int_{-\infty}^{\log\left(\frac{B/\alpha_2 - \alpha_1}{F_1(T)} \right)+\frac{\Sigma_{1,1}^{i,i}}{2}} d x_1 \frac{n(x_1;0,\Sigma_{1,1}^{i, i})n(D^{i,j}(B,x_1);0,(1-\rho^2) \Sigma_{2,2}^{j,j})}{B-\alpha_2F_1(T)e^{x_1- \Sigma_{1,1}^{i,i}/2}-\alpha_1 \alpha_2}
\end{align}
where $n(x;m,S)$ is the density function of a one-dimensional Gaussian random variable with mean $m$ and standard deviation $S$,
\begin{equation*}
D^{i,j}(B,x_1)=ln\left( \frac{B}{F_1(t)e^{x_1-\frac{\Sigma^{i,i}_{1,1}}{2}}+\alpha_1} - \alpha_2\right) - ln (F_2(t))+\frac{\Sigma^{j,j}_{2,2}}{2}-\rho x_1 \sqrt{\frac{\Sigma^{j,j}_{2,2}}{\Sigma^{i,i}_{1,1}}},
\end{equation*}
and $\Sigma^{i,j}_{h,k}=\sigma_h^i \sigma_k^j T$ for $h,k=1,2$ and $i,j=1,\dots N$. 
This follows from the fact that the density of the product 
$$B=S_1 S_2=(\xi_1+\beta_1e^{\mu_1 T})(\xi_2+\beta_2e^{\mu_2 T})$$
can be written as
\begin{equation}\label{eq:DensityBasket}
p_{B_T}(B)dB=\bQ(B_T\in dB)=\mathbb{E}[1_{\{B_T \in dB\}}]=\sum_{i,j=1}^N \lambda_1^i \lambda_2^j\mathbb{E}\left[1_{\{(\xi_1^i+\beta_1 e^{\mu_1 T}) (\xi_2^j +\beta_2 e^{\mu_2 T})\in dB\}}\right]
\end{equation}
where 
\begin{equation*}
\begin{array}{l l}
d\xi_1(t) = \mu_1\ \xi_1(t) dt + \sigma^{i}_1(t)\ \xi_1(t) dW_1(t)\\
d\xi_2(t) = \mu_2\ \xi_2(t) dt + \sigma^{j}_2(t)\ \xi_2(t) dW_2(t)
\end{array}
\end{equation*}
Now we focus on a single term in the summation (\ref{eq:DensityBasket}) and for simplicity we drop the superscript $i$, $j$ .
Calling $F_1(t)$, $F_2(t)$ the t-forward asset prices and defining $x_i=\ln \frac{\xi_i}{F_i(t)}+\frac{\Sigma_{i,i}}{2}$ we can rewrite the expectation as
\begin{align*}
&\int dx_1 dx_2 1_{\{(F_1(t)e^{x_1-\Sigma_{1,1}/2}+\alpha_1)(F_2(t)e^{x_2-\Sigma_{2,2}/2}+\alpha_2)\in dB \}}n(\underline{x};0,\Sigma)=\\
&\left( -\frac{d}{dB} \int_{D_B} dx_1 dx_2 n(\underline{x};0,\Sigma) \right)dB
\end{align*}
with $\alpha_i = \beta_i e^{\mu_i T}$, where $n(\underline{x};0,\Sigma)$ is the density of a bivariate normal distribution with mean equal to zero and covariance matrix $\Sigma$ defined as below
\begin{equation}
\Sigma=\begin{pmatrix}
\Sigma_{1,1} &    \rho \sqrt{\Sigma_{1,1} \Sigma_{2,2}}  \cr
\rho \sqrt{\Sigma_{1,1} \Sigma_{2,2}}&   \Sigma_{2,2}\cr
\end{pmatrix}
\end{equation}
Observing that $n(\underline{x};0,\Sigma)=n(x_1;0, \Sigma_{1,1})n(x_2-\rho x_1 \sqrt{\Sigma_{2,2}/\Sigma_{1,1}};0,(1-\rho^2)\Sigma_{2,2})$, integrating with respect to $x_2$ and replacing in (\ref{eq:DensityBasket}) we obtain
\begin{equation*}
p_{B_T}(B)=\sum_{i,j=1}^N \lambda_1^i \lambda_2^j \int_{-\infty}^{\log\left(\frac{B/\alpha_2 - \alpha_1}{F_1(T)} \right)+\frac{\Sigma_{1,1}^{i,i}}{2}} dx_1 \frac{n(x_1;0,\Sigma_{1,1}^{i, j})n(D^{i,j}(B,x_1);0,(1-\rho^2) \Sigma_{2,2}^{i,j})}{B-\alpha_2F_1(T)e^{x_1- \Sigma_{1,1}^{i,j}}-\alpha_1 \alpha_2}
\end{equation*}
from which equation (\ref{eq:PriceBaske}) is easily derived.

\section{Comparing correlation skews in shifted MVMD and SCMD}\label{sec:Comparison}
The aim of this section is to compare the shifted MVMD and the shifted SCMD models in terms of implied correlation, analysing their performance in reproducing triangular relationships.

\subsection{Numerical case study with cross FX rates}
Specifically, we consider the exchanges $S_1=$ USD/EUR, $S_2=$ EUR/JPY under a shifted MUVM model with $2$ components 

\begin{align*}
S_1(t)=X_1(t)+\beta_1e^{(r^{\textup{\EUR{}}}-r^{\$})t}\\
S_2(t)=X_2(t)+\beta_2 e^{(r^{Y}-r^{\textup{\EUR{}}})t}
\end{align*}
with
\begin{align*}
dX_1(t)=(r^{\textup{\EUR{}}}-r^{\$})X_1(t)dt+\sigma_1^{I_1}(t) X_1(t)dW_t^{1,\textup{\EUR{}}}\\
dX_2(t)=(r^{Y}-r^{\textup{\EUR{}}})X_2(t)dt+\sigma_2^{I_2}(t) X_2(t)dW_t^{2,Y}
\end{align*}

where $r^{\textup{\EUR{}}}$, $r^{\$}$, $r^{Y}$ are the euro, dollar and yen interest rates, respectively, and $\sigma_1^{I_1}(t)$, $\sigma_2^{I_2}(t)$ are as in equation (\ref{edsuncertain_volatily_model}).
$W_t^{1,\textup{\EUR{}}}$  and $W_t^{2,Y}$ indicate that we are considering the dynamics of $S_1$ and of $S_2$, each under its own domestic measure, that is the euro in the case of $S_1$ and the yen in the case of $S_2$.

We calibrate $S_1$ and $S_2$ independently, each to its own volatility curve, using 2 components and minimizing the squared percentage difference between model and market implied volatilities. Then, we look at the product $S_1 S_2$, representing the cross exchange USD/JPY, and we check whether the model is able to reproduce the cross smile consistently with the smiles of the single assets. In particular, we find $\rho$ that minimizes the squared percentage difference between implied volatilities from options on the basket $S_3=S_1 S_2$ under the shifted MVMD model (the shifted SCMD model) and market implied volatilities. In other words, we look at the implied correlations under the shifted MVMD model and the shifted SCMD model.

When performing calibration on $S_3$, we express both the dynamics of $X_1$ and of $X_2$ under the yen:

\begin{align*}
dX_1(t)&=(r^{\textup{\EUR{}}}-r^{\$}-\rho \sigma_1^{I_1}(t) \sigma_2^{I_2}(t))X_1(t)dt+\sigma_1^{I_1} X_1(t)dW_t^{1,Y}\\
dX_2(t)&=(r^{Y}-r^{\textup{\EUR{}}})X_2(t)dt+\sigma_2^{I_2}(t) X_2(t)dW_t^{2,Y},
\end{align*}

and then we calculate prices of options on 
$$S_3(t)=(X_1(t)+\beta_1e^{(r^{\textup{\EUR{}}}-r^{\$})t})(X_2(t)+\beta_2 e^{(r^{Y}-r^{\textup{\EUR{}}})t}).$$

All the data for our numerical experiments are downloaded from a Bloomberg terminal. 
We start by considering data from 19th February 2015. The initial values of $S_1,S_2$ are $S_1(0)=0.878$, $S_2(0)=135.44$. First we calibrate $S_1$ and $S_2$ using implied volatilities from options with maturity of $6$ months. Denoting 
\begin{align*}
\eta_1=\left(\sqrt{\frac{\int_0^T \sigma_1^1(s)^2 ds}{T}}, \sqrt{\frac{\int_0^T \sigma_1^2(s)^2 ds}{T}}\right) \\
\eta_2=\left( \sqrt{\frac{\int_0^T \sigma_2^1(s)^2 ds}{T}}, \sqrt{\frac{\int_0^T \sigma_2^2(s)^2 ds}{T}}\right)
\end{align*}
the $T$-term volatilities of the instrumental processes of $S_1$ and $S_2$, respectively, 
\begin{align*}
\lambda_1=(\lambda_1^1,\lambda_1^2),\\
\lambda_2=(\lambda_2^1,\lambda_2^2)
\end{align*}
the vector of probabilities of each component and $\beta_1$, $\beta_2$ the shift parameters, we obtain
$$\eta_1=(0.1952,0.0709),\ \lambda_1=(0.1402,0.8598),\ \beta_1=0.00068$$
for the asset $S_1$ and
$$\eta_2=(0.1184,0.0962),\ \lambda_2=(0.2735,0.7265),\ \beta_2=0.9752$$
for the asset $S_2$. Then, we perform a calibration on the cross product $S_3=$USD/JPY using volatilities from call options with maturity of $6$ months, finding the values: 
$$\rho_{MVMD}(6M)=-0.6015$$
for the shifted MVMD model and 
$$\rho_{SCMD}(6M)=-0.5472$$
for the shifted SCMD model. The higher value (in absolute terms) of the correlation parameter in the shifted MVMD model is due to higher state dependence in the diffusion matrix with respect to the shifted SCMD model. This is partly related to Proposition \ref{Prop:LocalCorrelation}. In other words, in order to achieve the same local correlation as in the shifted SCMD model, the shifted MVMD model needs a higher absolute value of $\rho$.

The corresponding prices and implied volatilities are plotted in Figure \ref{Fig:Basket6M-19Feb} whereas Table \ref{Table:Basket6M-19Feb} reports the absolute differences between market and model values corresponding to a few strikes. The reported plot shows that the shifted MVMD model is better at reproducing market prices than the shifted SCMD model. What is very remarkable in this example is that the shifted MVMD fits the whole correlation skew with just one value of $\rho$.

As a second numerical experiment we repeat the calibration using prices with maturity of $9$ months. Specifically, we first calibrate $S_1$ and $S_2$ obtaining the values 
$$\eta_1=(0.2236,0.0761),\ \lambda_1=(0.0262,0.9738),\ \beta_1=0.0100$$
for the asset $S_1$ and
$$\eta_2=(0.1244,0.0497),\ \lambda_2=(0.7584,0.2416),\ \beta_2=0.7856$$
for the asset $S_2$. For $S_2$, we observe that the higher volatility now has the highest probability as opposed to the results found for 6 months options. Then, we perform a calibration on the cross product $S_3=$USD/JPY using volatilities from call options with maturity of $9$ months, finding the values: 
$$\rho_{MVMD}(9M)=-0.6199$$
for the shifted MVMD model and 
$$\rho_{SCMD}(9M)=-0.5288$$
for the shifted SCMD model. These values are comparable with those found for 6 months options. This shows that the model is quite consistent. 

The corresponding prices and implied volatilities are shown in Figure \ref{Fig:Basket9M-19Feb} whereas Table \ref{Table:Basket9M-19Feb} reports some absolute differences between model and market values.
Overall, also in this case the shifted MVMD model outperforms the shifted SCMD in terms of ability to reproduce market prices on the cross product.

\begin{figure}[h!]
\begin{minipage}[b]{\linewidth}
\centering
\includegraphics[scale=0.34]{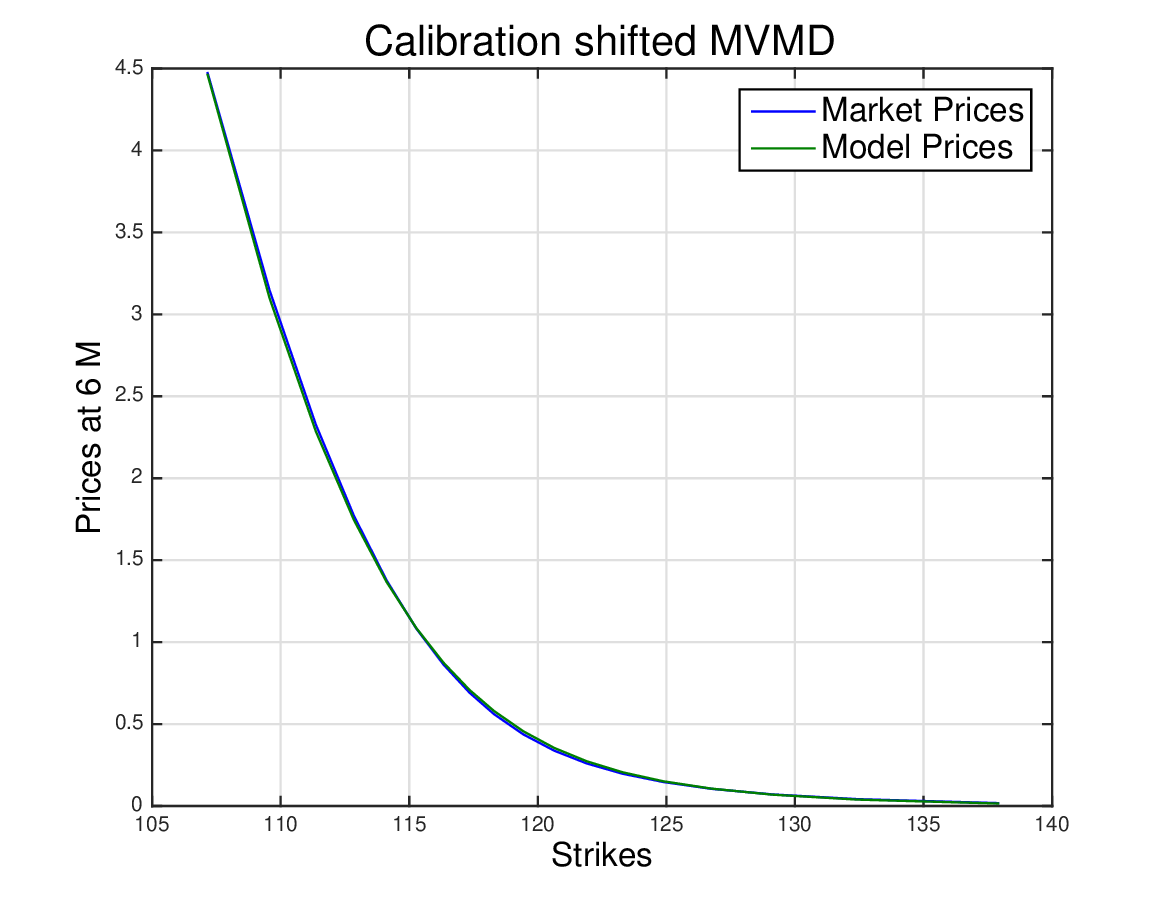}
\centering
\includegraphics[scale=0.34]{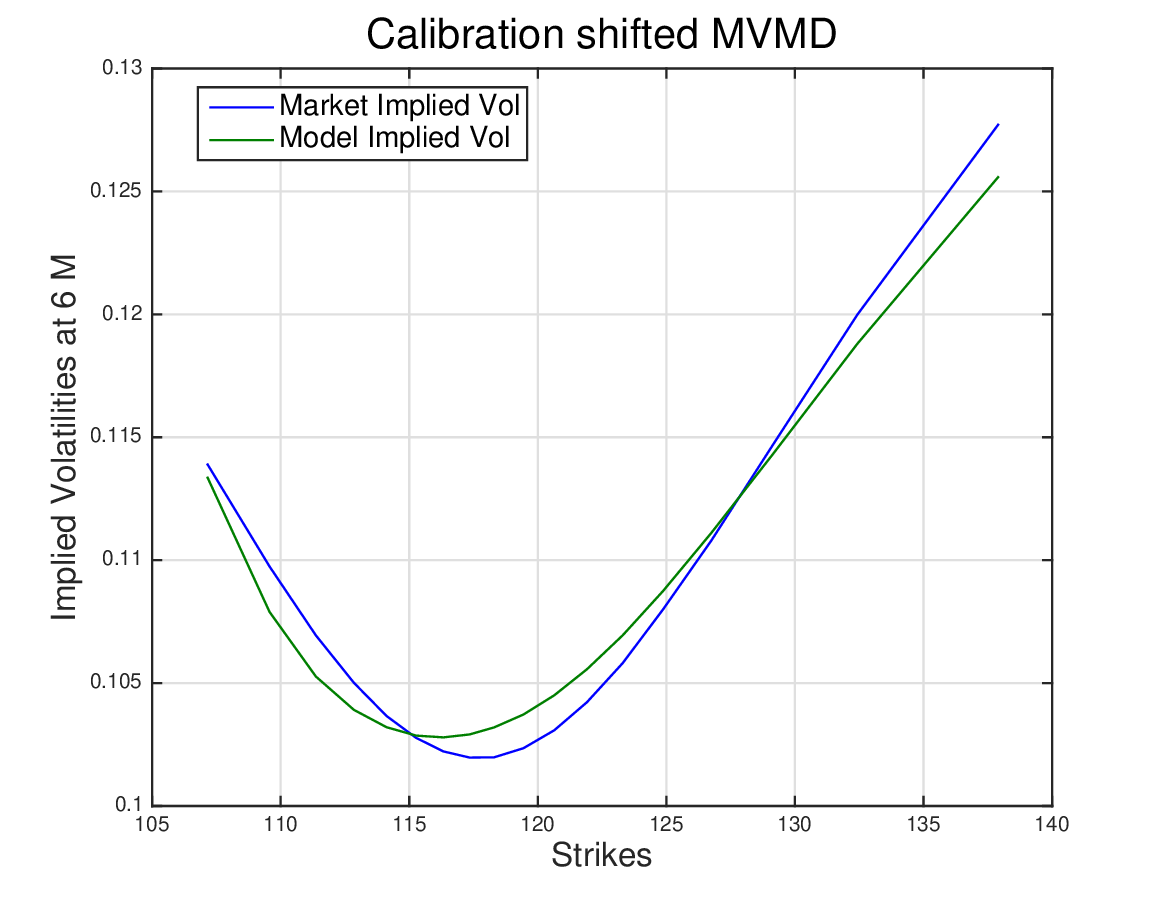}
\end{minipage}
\begin{minipage}[b]{\linewidth}
\centering
\includegraphics[scale=0.34]{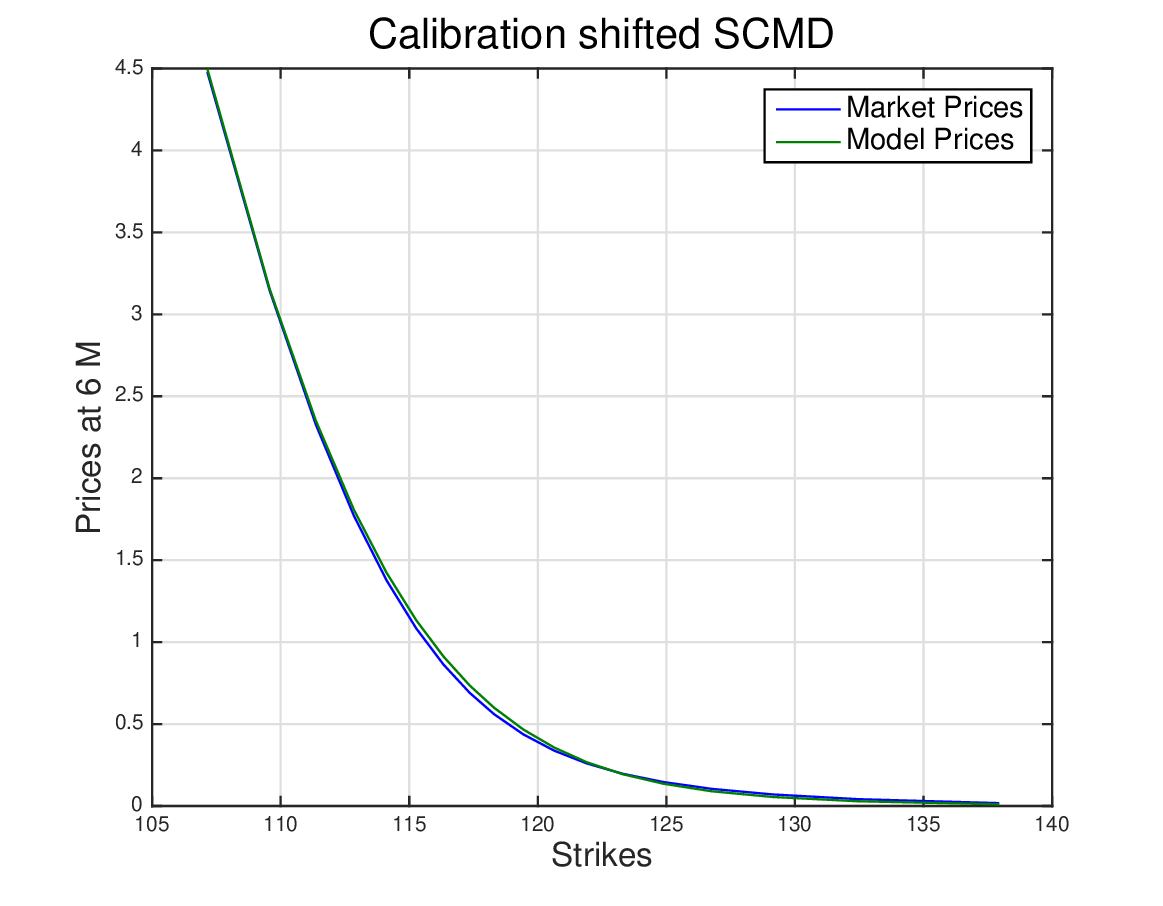}
\centering
\includegraphics[scale=0.34]{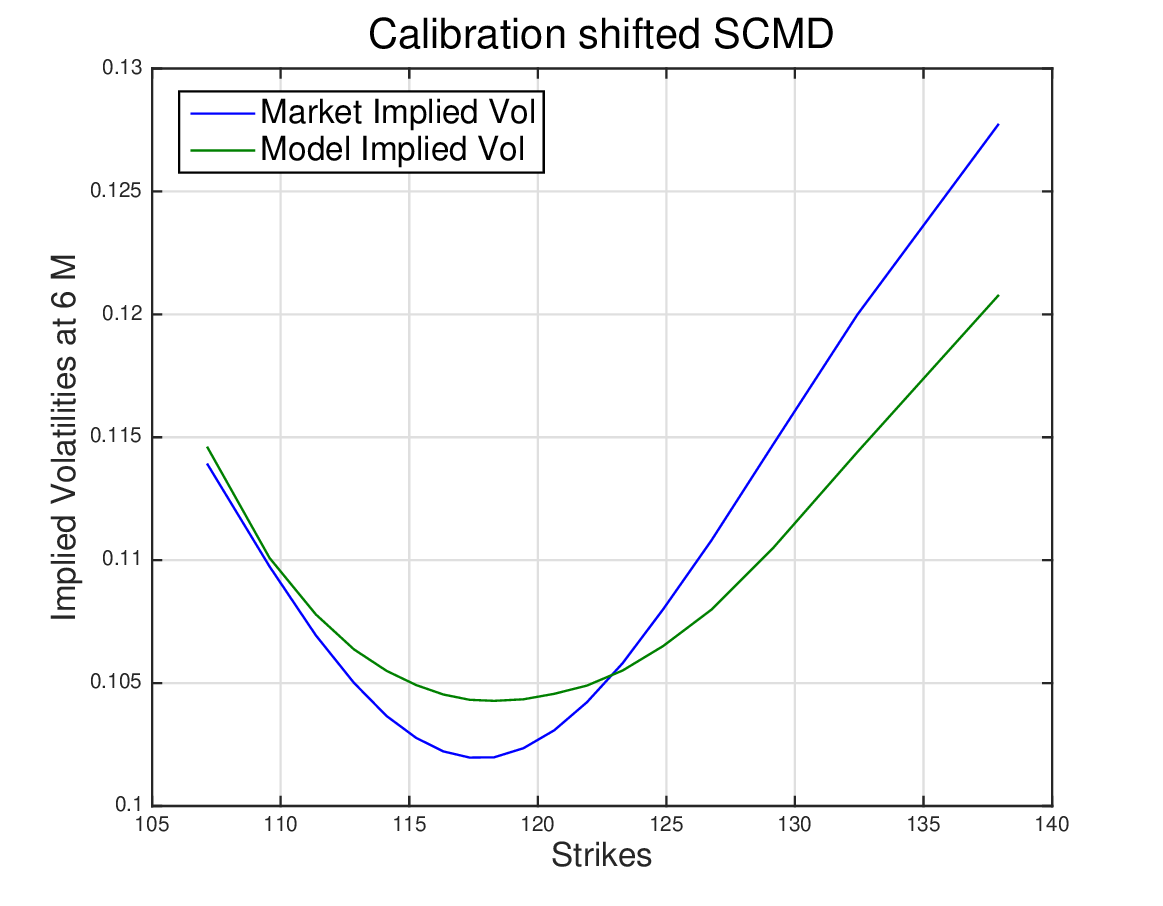}
\end{minipage}
\caption{\textit{Calibration on $6$ months options, relative to 19 February 2015. The implied correlation is $\rho=-0.6015$ for the shifted MVMD model (top) and $\rho=-0.5472$ for the shifted SCMD model (bottom). }}
\label{Fig:Basket6M-19Feb}
\end{figure}
\begin{table}[htb]
\begin{center}
\begin{tabular}{|c|c|c|}
\hline
\multicolumn {3}{|c|}{$T = 6$ Months}\\
\hline
K & Shifted  MVMD  & Shifted SCMD \\
\hline
107.16 & 0.0107 & 0.0188 \\ 
114.12 & 0.0095 &   0.0463  \\
118.3 & 0.019 &  0.0394 \\ 
124.88 & 0.0045 & 0.0104 \\
137.9 &0.0026 & 0.0073  \\
\hline 
\end{tabular}\\
\vspace{0.5cm}
\begin{tabular}{|c|c|c|}
\hline
\multicolumn {3}{|c|}{$T = 6$ Months}\\
\hline
K & Shifted MVMD  & Shifted SCMD \\
\hline
107.16 & 0.0004 & 0.0007\\ 
114.12 & 0.0004    & 0.0018 \\
118.3 & 0.0011  & 0.0023\\ 
124.88 & 0.0006 & 0.0015\\
137.9 & 0.0022 & 0.007 \\
\hline 
\end{tabular}
\end{center}
\caption{\textit{Calibration on $6$ months options, relative to 19 February 2015. The tables report absolute differences between market and model prices (top) and absolute differences between market and model implied volatilities (bottom).}}
\label{Table:Basket6M-19Feb}
\end{table}

\begin{figure}[h!]
\begin{minipage}[b]{\linewidth}
\centering
\includegraphics[scale=0.34]{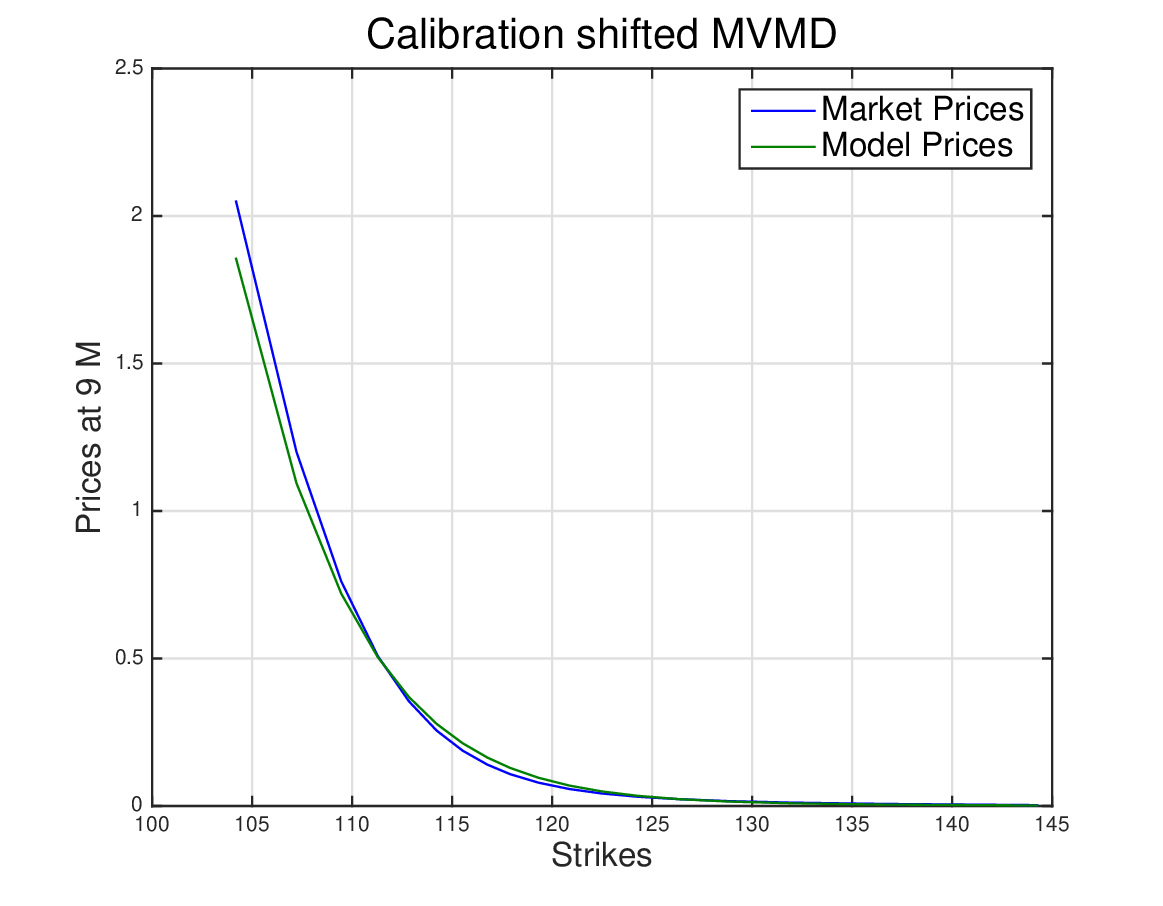}
\centering
\includegraphics[scale=0.34]{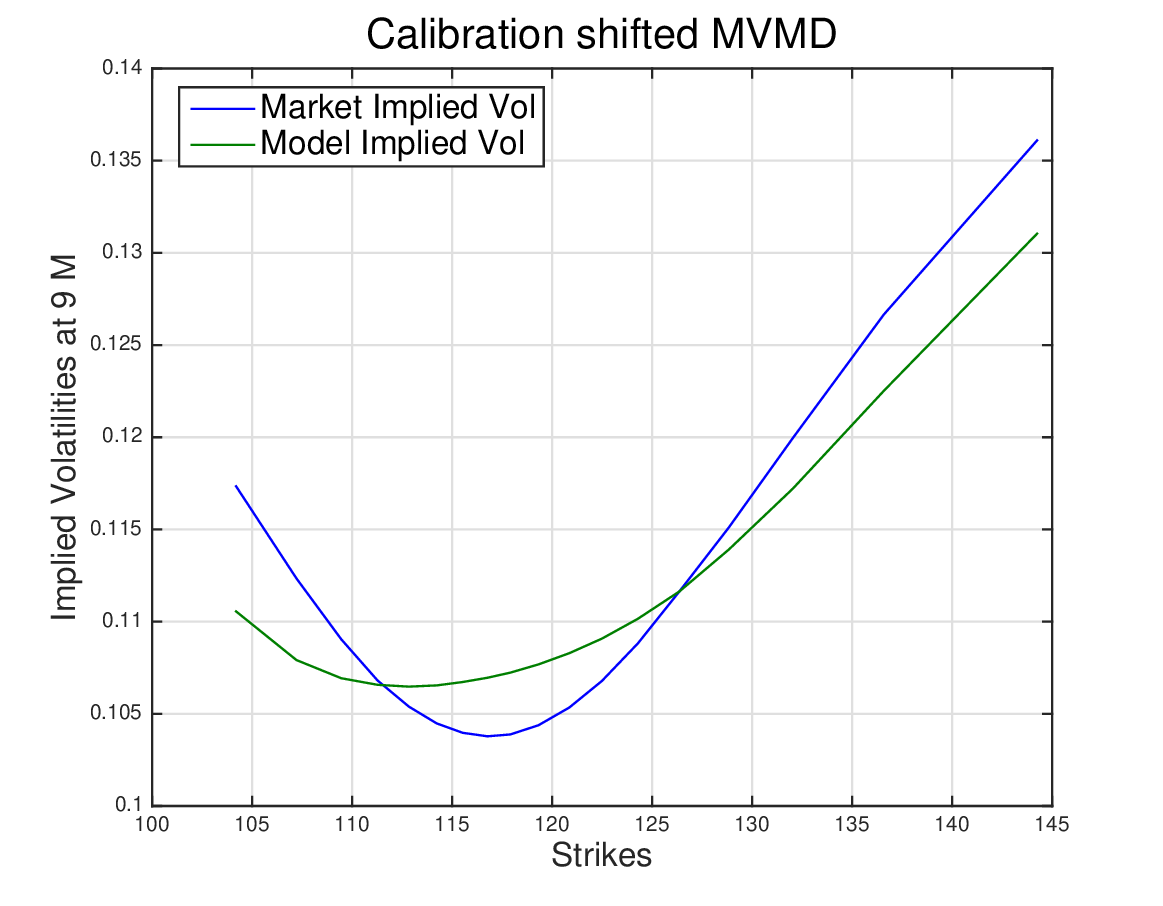}
\end{minipage}
\begin{minipage}[b]{\linewidth}
\centering
\includegraphics[scale=0.34]{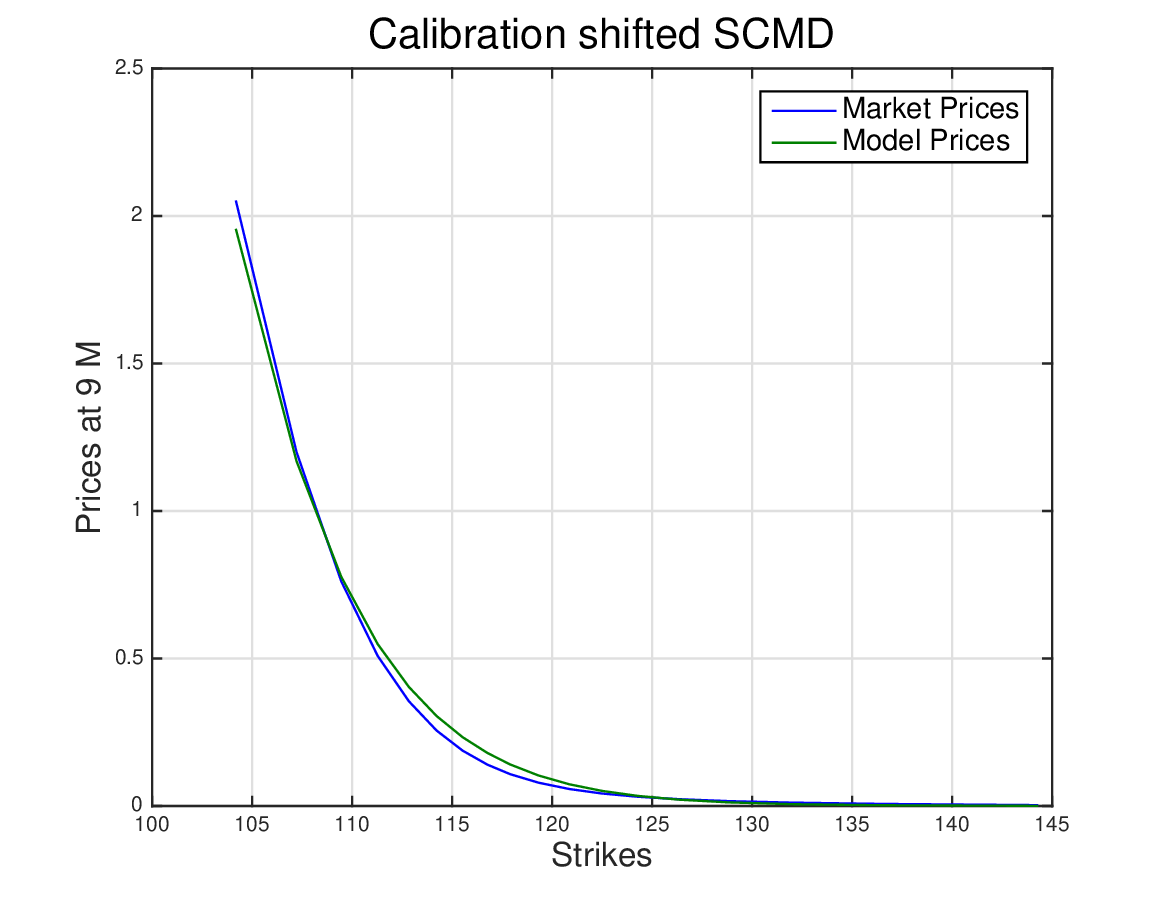}
\centering
\includegraphics[scale=0.34]{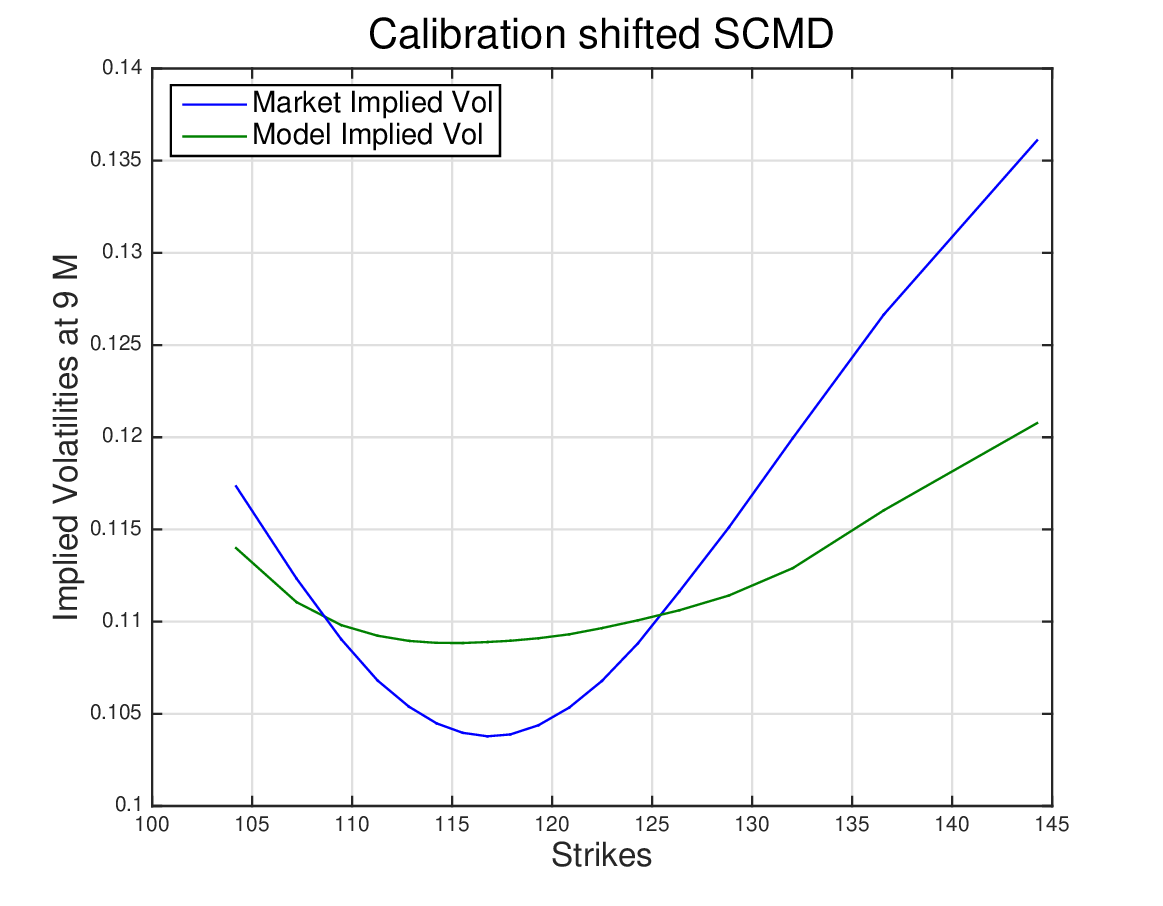}
\end{minipage}
\caption{\textit{Calibration on $9$ months options, relative to 19 February 2015. The implied correlation is $\rho=-0.6199$ for the shifted MVMD model (top) and $\rho=-0.5288$ for the shifted SCMD model (bottom), which are comparable with the values obtained using $6$ months options.}}
\label{Fig:Basket9M-19Feb}
\end{figure}
\begin{table}[htb]
\begin{center}
\begin{tabular}{|c|c|c|}
\hline
\multicolumn {3}{|c|}{$T = 9$ Months}\\
\hline
K & Shifted MVMD  & Shifted SCMD \\
\hline
104.2 & 0.2098 & 0.0961 \\ 
112.84 & 0.0116 &   0.0485  \\
117.91 & 0.0223 &  0.0329 \\ 
122.49 & 0.0088 & 0.0089 \\
144.25 & 0.0013 & 0.0025  \\
\hline 
\end{tabular}\\
\vspace{0.5cm}
\begin{tabular}{|c|c|c|}
\hline
\multicolumn {3}{|c|}{$T = 9$ Months}\\
\hline
K &  Shifted MVMD  & Shifted SCMD \\
\hline
104.2 & 0.0073 & 0.0034 \\ 
112.84 & 0.0008 &   0.0024  \\
117.91 & 0.0035 &  0.005 \\ 
122.49 & 0.0028 & 0.0029 \\
144.25 & 0.005 & 0.01  \\
\hline 
\end{tabular}
\end{center}
\caption{\textit{Calibration on $9$ months options, relative to 19 February 2015. The tables report absolute differences between market and model prices (top) and absolute differences between market and model implied volatilities (bottom).}}
\label{Table:Basket9M-19Feb}
\end{table}

\section{Introducing random correlations in the mixture dynamics}\label{Comparison-RandomCorrelation}
A single correlation parameter $\rho$ may not be enough to fit prices on the cross asset. To overcome this, we can allow for random correlations between the single assets in the shifted MUVM model (\ref{eq:ShiftedMUVM}). Specifically,

\begin{equation}\label{eq:ShiftedMUVMRandomCorrelations}
\begin{array}{l l}
d S_1(t) = \mu_1\ S_1(t) dt + \sigma^{I_1}_1(t)\ (S_1(t)-\beta_1 e^{\mu_1 t}) dW^{I_1}_1(t)\\
d S_2(t) = \mu_2\ S_2(t) dt + \sigma^{I_2}_2(t)\ (S_2(t)-\beta_2 e^{\mu_2 t}) dW^{I_2}_2(t)
\end{array}
\end{equation}
where the Brownian motions $W_1^{I_1}$, $W_2^{I_2}$ now have correlation $\rho^{I_1,I_2}$. The correlation parameter will therefore assume the value $\rho^{h,k}$ in correspondence with a couple $(\sigma^{h}_1, \sigma^{k}_2)$ with probability $\lambda_h \lambda_k$. 

\begin{theorem}
The shifted MUVM model with uncertain correlation parameter has, as Markovian projection, a shifted MVMD model solution of the SDE (\ref{edsShiftedMVMD}) but with equation (\ref{V}) transformed into 
\begin{equation}\label{V-random}
{V}^{k_1,...,k_n}(t) = 
 \left[\sigma_i^{k_i}(t)\ \rho^{k_i,k_j}_{i,j}\
\sigma_j^{k_j}(t)\right]_{i,j = 1,...,n}.
\end{equation}
\end{theorem}
\begin{proof}
The Markovian projection property can be easily shown by an application of Gy\"{o}ngy's lemma, in a similar was as in the proof of Theorem \ref{Theo:MarkovianProjectionShiftedModels}.
\end{proof}

In other words, the correlation between two generic instrumental processes $Y_i^k$, $Y_j^h$ will depend not only on the assets $S_i$, $S_j$, but will correspond to a specific choice of the instrumental processes $Y_i^k$, $Y_j^h$ themselves.

\subsection{Cross FX rates study for shifted MVMD with random correlations} 
As a numerical illustration we performed on the shifted MVMD model the same experiment as in Section \ref{sec:Comparison}. We used 6 months options from 7th September 2015. The initial values of the single FX rates are $S_1(0)=0.8950$, $S_2(0)=133.345$. 
In this case the calibration of the shifted SCMD model is much worse, to the point that there is no value of $\rho$ that can fit any of the prices obtained through this model. On the other hand, in the case of the shifted MVMD model, in particular when introducing random correlations, the fit leads to quite good results. 
As in the previous cases we independently calibrate $S_1$=USD/EUR and $S_2$=EUR/JPY on the corresponding implied volatilities obtaining
$$\eta_1=(0.1803,0.0916),\ \lambda_1=(0.0274,0.9726),\ \beta_1=0.0128$$
$$\eta_2=(0.1230,0.0501),\ \lambda_2=(0.6575,0.3425),\ \beta_2=0.1867$$

and then we look at the cross exchange rate $S_3=S_1 S_2$=USD/JPY. When performing calibration using a shifted MVMD model with one correlation parameter only, we obtain  
$$\rho=-0.6147,$$
whereas when using random correlations, we have
$$\rho^{1,1}=-0.8717,\ \rho^{1,2}=-0.1762,\ \rho^{2,1}=-0.6591,\ \rho^{2,2}=	-0.2269. $$

The corresponding plots are shown in Figure \ref{Fig:BasketRandomCorrelation6M-Sept2015}, in connection with Table \ref{Table:BasketRandomCorrelation6M-Sept2015}. In this case we also see that using random correlations improves the fit with respect to the case with a single correlation parameter. Moreover, computing the expectation and the standard deviation for the random correlation under the risk-neutral measure $\bQ$, we obtain
$$\bE^{\bQ}( \rho^{i,j})=-0.5144$$
$$ \mbox{Std}^{\bQ}(\rho^{i,j})=0.2105$$
satisfying $\vert \bE^{Q}(\rho^{i,j}) - \rho \vert < \frac{ \mbox{Std}^Q(\rho^{i,j})}{2}$. In other words, the absolute difference between the $\bQ$-expected random correlation and the deterministic correlation is smaller than half the $\bQ$-standard deviation. This means that the random correlation is on average not that far from the deterministic value.

\begin{figure}[h!]
\begin{minipage}[b]{\linewidth}
\centering
\includegraphics[scale=0.34]{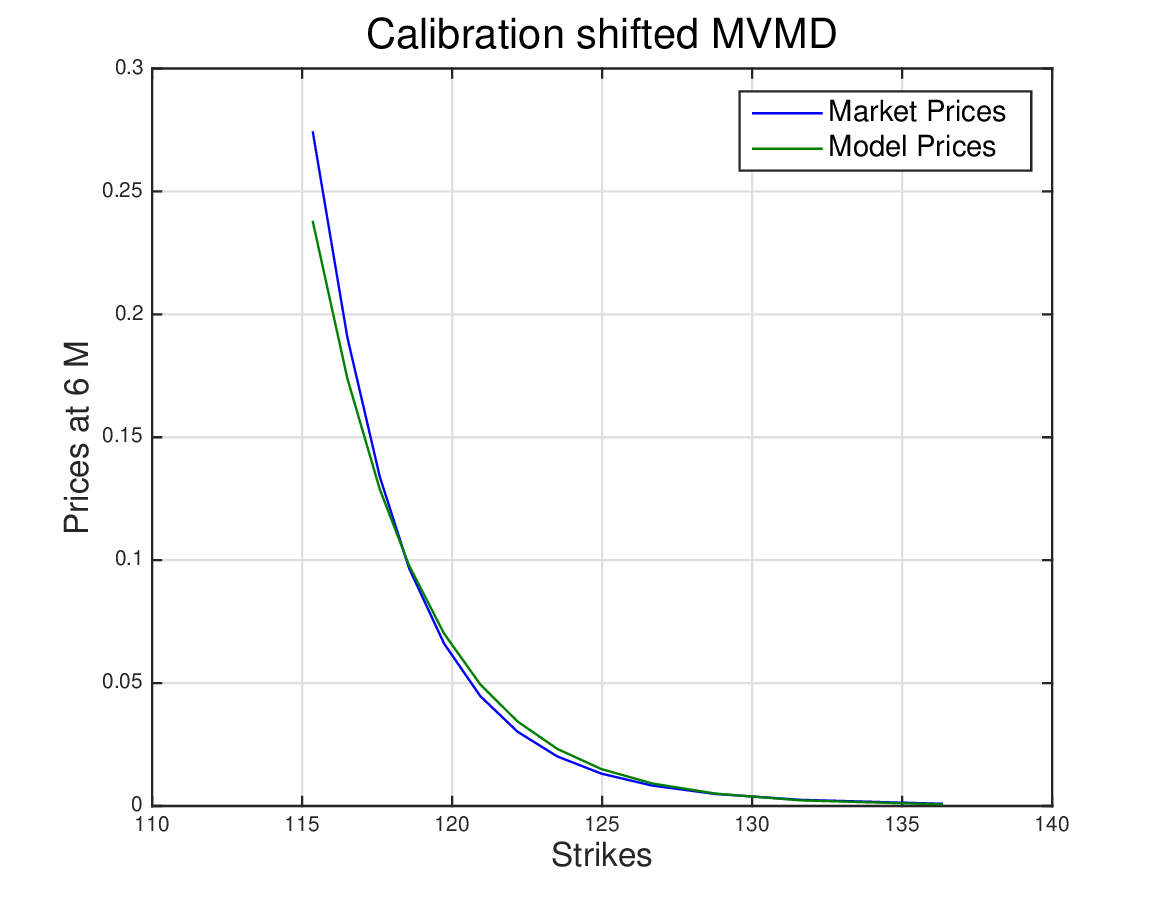}
\centering
\includegraphics[scale=0.34]{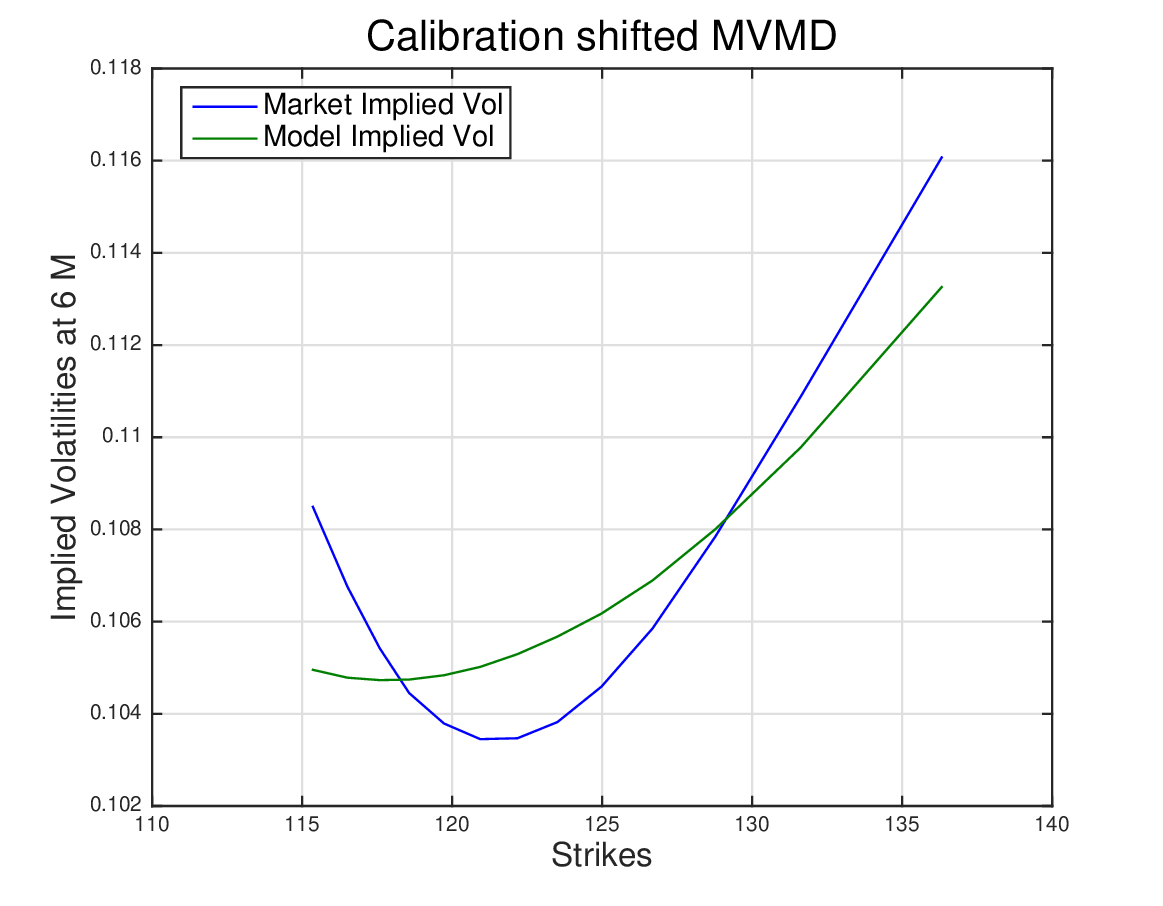}
\end{minipage}
\begin{minipage}[b]{\linewidth}
\centering
\includegraphics[scale=0.34]{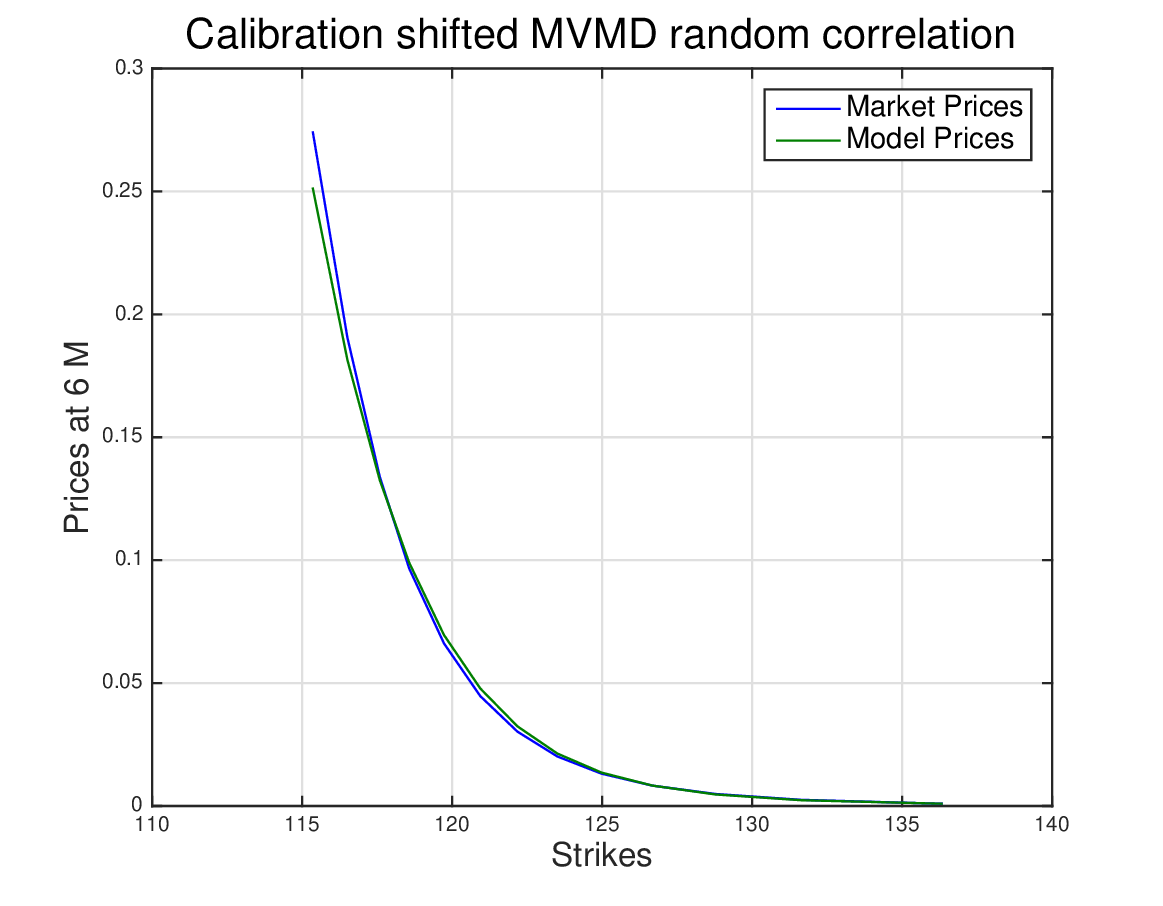}
\centering
\includegraphics[scale=0.34]{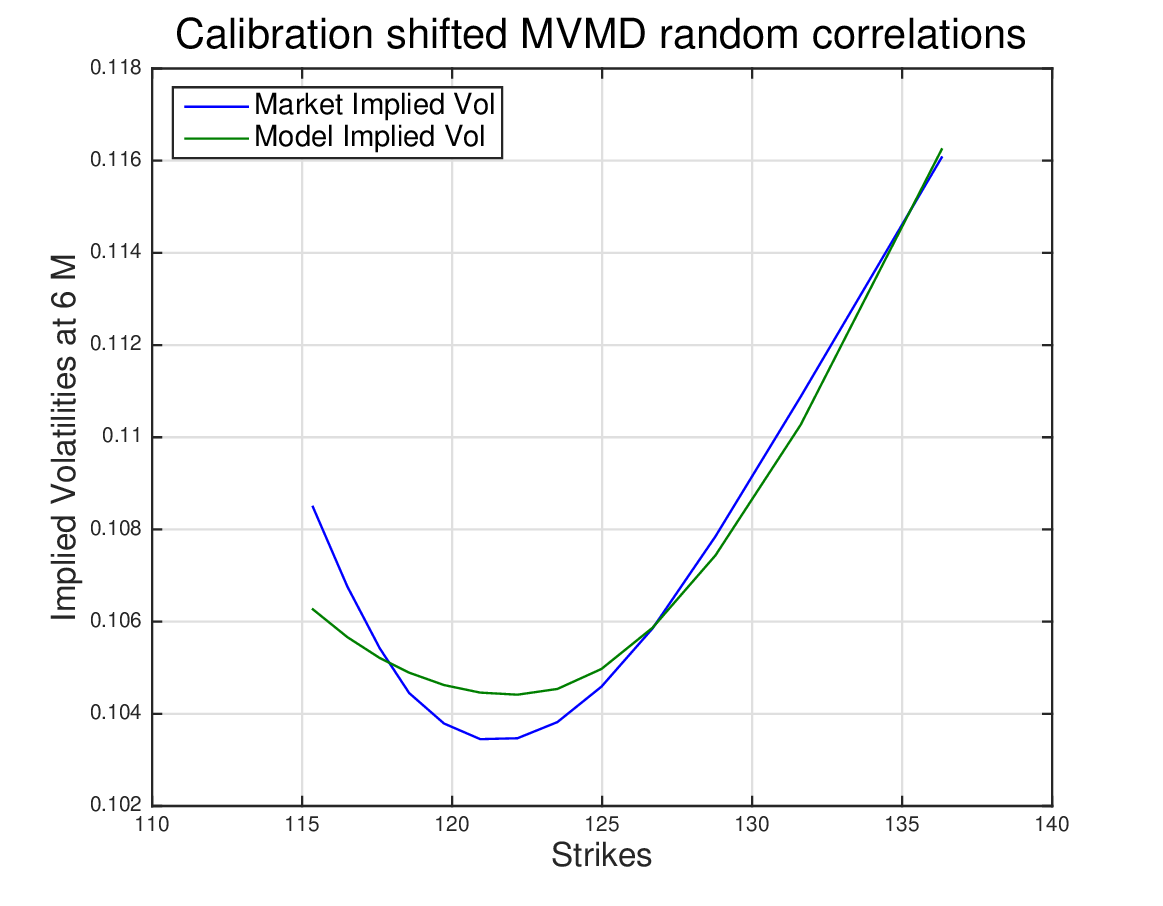}
\end{minipage}
\caption{\textit{Calibration of the MVMD model on $6$ months options, relative to data from 7th September 2015. The calibration using one single correlation parameter is shown in the top part which corresponds to a fitted value equal to $\rho=-0.6147$. In the bottom, calibration using the MVMD model with random correlations is presented. The corresponding fitted correlations are $\rho^{1,1}=-0.8717,\ \rho^{1,2}=-0.1762,\ \rho^{2,1}=-0.6591,\ \rho^{2,2}=	-0.2269. $
}}
\label{Fig:BasketRandomCorrelation6M-Sept2015}
\end{figure}
\begin{table}[htb]
\begin{center}
\begin{tabular}{|c|c|c|}
\hline
\multicolumn {3}{|c|}{$T = 6$ Months}\\
\hline
K &  Shifted MVMD  & Shifted MVMDRC \\
\hline
115.36 & 0.0359 & 0.02228 \\ 
118.57 & 0.0016 &   0.0024  \\
122.18 & 0.0042 &  0.0022 \\ 
126.68 & 0.0008 & $1.18*10^{-5}$ \\
136.32 & 0.0003 & $1.94*10^{-5}$  \\
\hline 
\end{tabular}\\
\vspace{0.5cm}
\begin{tabular}{|c|c|c|}
\hline
\multicolumn {3}{|c|}{$T = 6$ Months}\\
\hline
K & Shifted MVMD  & Shifted MVMDRC \\
\hline
115.36 & 0.0036 & 0.0022 \\ 
118.57 & 0.0002 &   0.0004  \\
122.18 & 0.0018 &  0.0009 \\ 
126.68 & 0.0011 & $1.159*10^{-5}$ \\
136.32 & 0.0025 & 0.0002  \\
\hline 
\end{tabular}
\end{center}
\caption{\textit{Calibration on $6$ months options, relative to 7th September 2015. The tables report absolute differences between market and model prices (top) and absolute differences between market and model implied volatilities (bottom).}}
\label{Table:BasketRandomCorrelation6M-Sept2015}
\end{table}

Finally, we repeat the same experiment using options with maturity of 9 months. We find:
$$\eta_1=(0.2073,0.0936),\ \lambda_1=(0.0012,0.9988),\ \beta_1=0.0216$$
$$\eta_2=(0.1573,0.0689),\ \lambda_2=(0.3563,0.6437),\ \beta_2=0.1288.$$

When looking at the cross product $S_1 S_2$=USD/JPY, we obtain 
$$\rho=-0.7488$$
in the case of one single correlation parameter, and 
$$\rho^{1,1}=-0.8679,\ \rho^{1,2}=-0.2208,\ \rho^{2,1}=-0.8303,\ \rho^{2,2}=	-0.3270$$
in the case where random correlations are introduced. Corresponding plots and absolute differences between market and model prices/implied volatilities can be found in Figure \ref{Fig:BasketRandomCorrelation9M-Sept2015} and Table \ref{Table:BasketRandomCorrelation9M-Sept2015}, which show that the shifted MVMD model with random correlations outperforms the constant-deterministic correlation model in this case as well. 

The values of expected random correlation and standard deviation under the $\bQ$ measure are 
$$\bE^{\bQ}( \rho^{i,j})=-0.5063$$
$$\mbox{Std}^{\bQ}(\rho^{i,j})=0.2411.$$
With respect to the case of $6$ months options, we observe a movement of the $\bQ$-expected random correlation away from the constant correlation. Moreover, if we look at the terminal correlations, that is the correlation between $S_1(T)$ and $S_2(T)$, for $T=9$ months, we obtain 
$$\hat \rho(9M)=-0.6894$$
in case $\rho$ is deterministic and 
$$\hat \rho_{\mbox{random}}(9M)=-0.5596$$
in case $\rho$ is random. As a final observation, we remark that in case $\rho$ is constant, an application of Schwartz's inequality shows that the absolute value of the terminal correlation is always smaller than the absolute value of the instantaneous correlation, as verified by the results above. One may wonder whether the same inequality holds in case of random correlations, if we substitute the instantaneous value with the mean of the random correlations. In this case it is not possible to use Schwartz's inequality as we did before and, indeed, the results obtained show that the inequality does not hold, at least not for the example considered above.

\begin{figure}[h!]
\begin{minipage}[b]{\linewidth}
\centering
\includegraphics[scale=0.34]{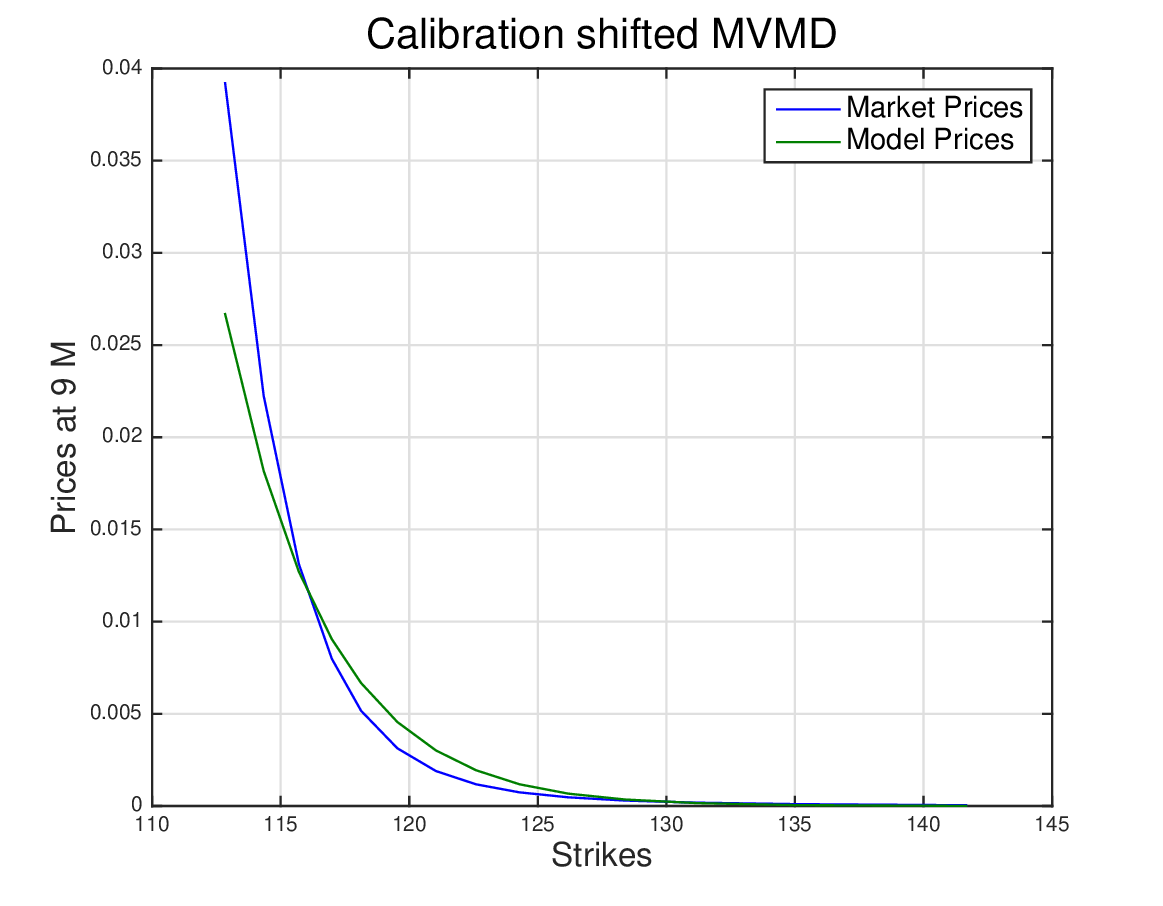}
\centering
\includegraphics[scale=0.34]{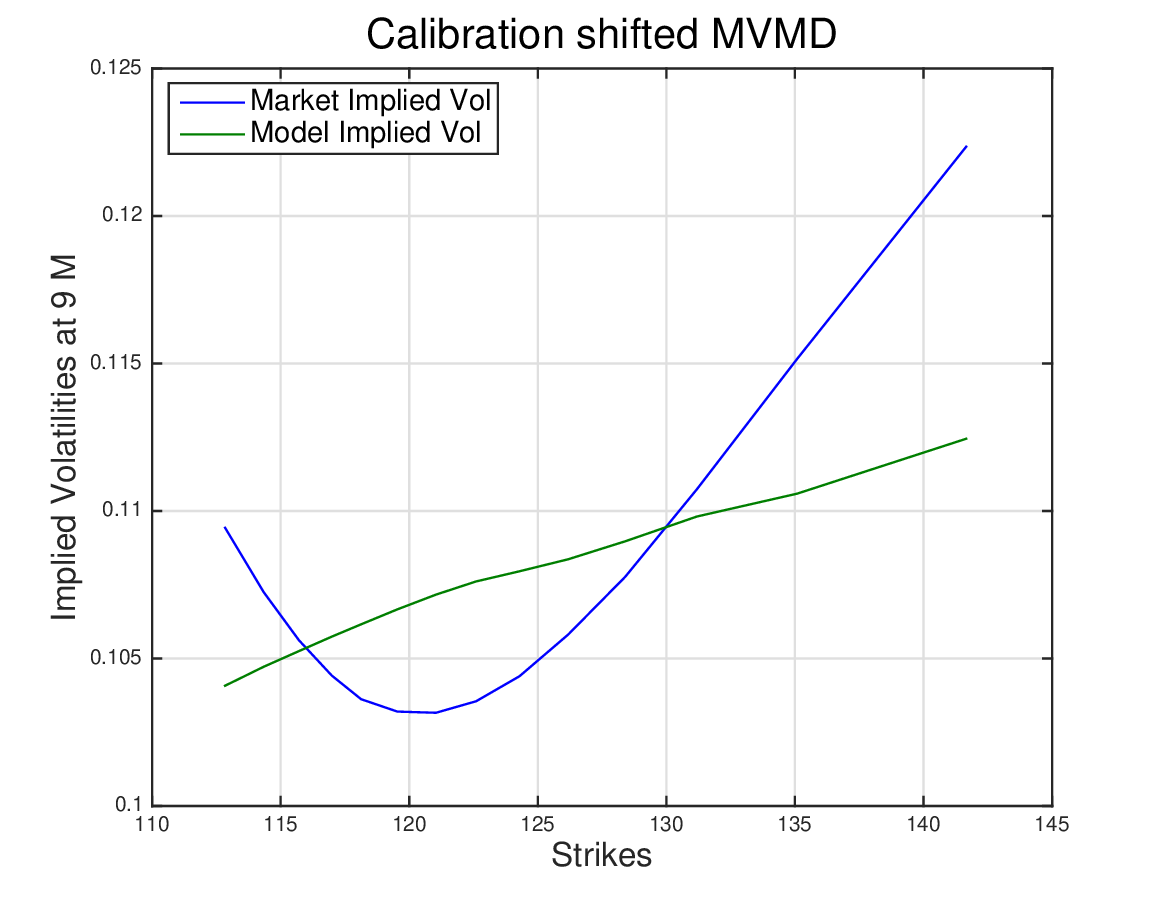}
\end{minipage}
\begin{minipage}[b]{\linewidth}
\centering
\includegraphics[scale=0.34]{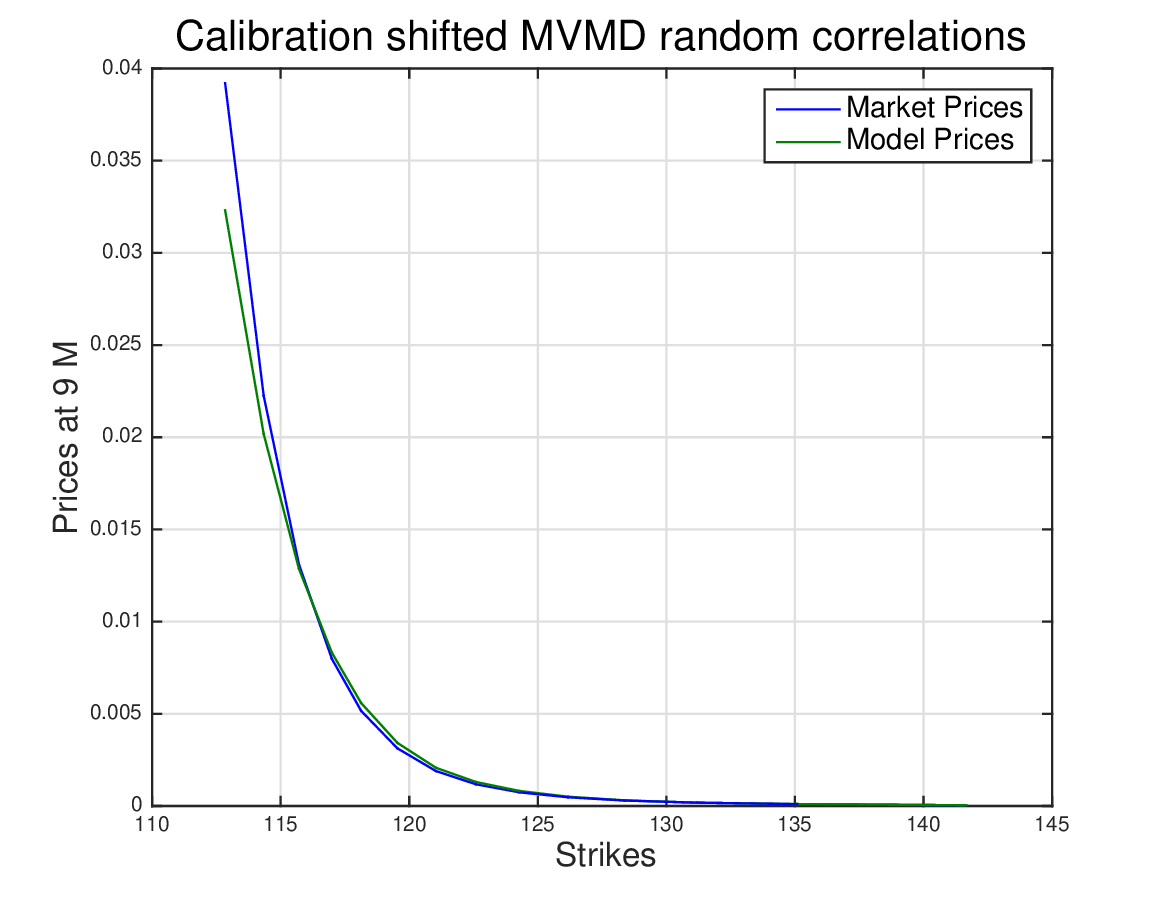}
\centering
\includegraphics[scale=0.34]{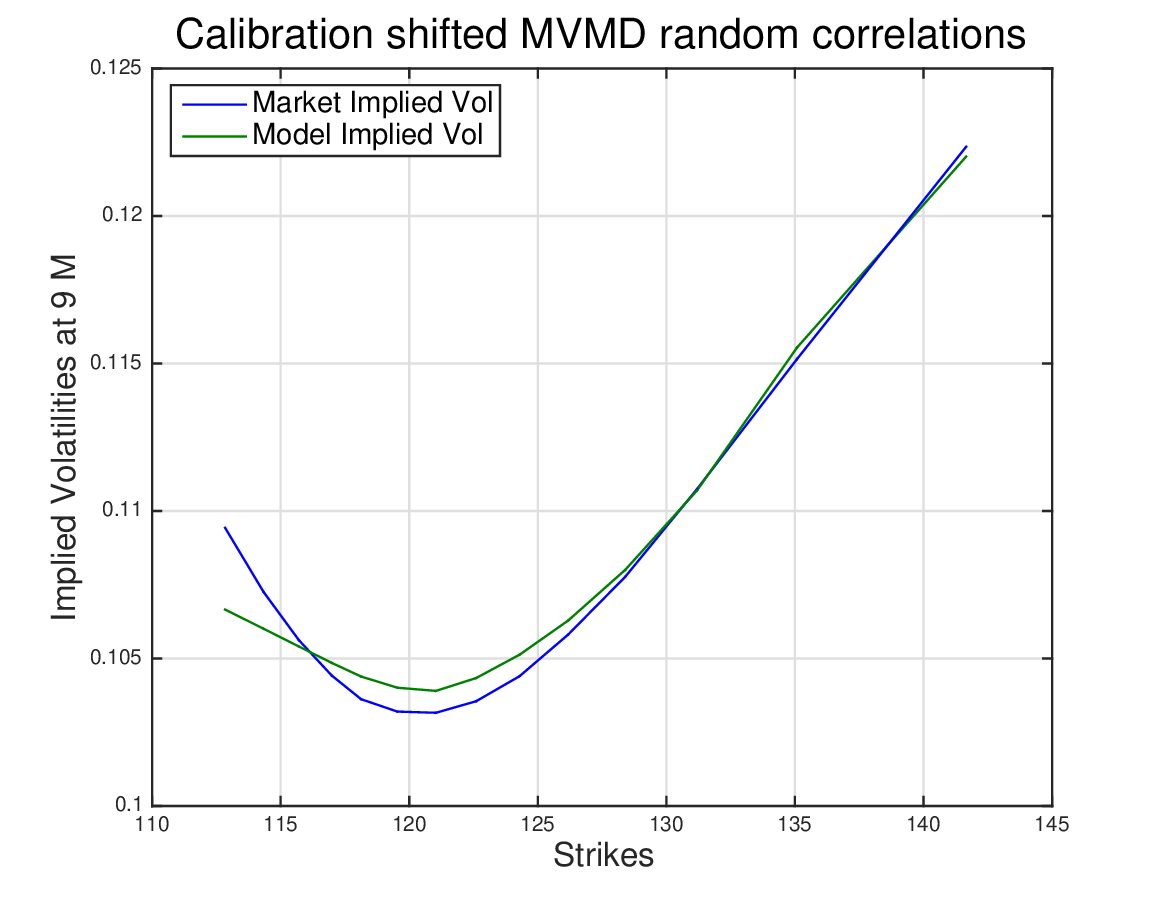}
\end{minipage}
\caption{\textit{Calibration of the MVMD model on $9$ months options, relative to data from 7th September 2015. The calibration using one single correlation parameter is shown in the top part which corresponds to a fitted value equal to $\rho=-0.7488$. In the bottom, calibration using the MVMD model with random correlations is presented. The corresponding fitted correlations are $\rho^{1,1}=-0.8679,\ \rho^{1,2}=-0.2208,\ \rho^{2,1}=-0.8303,\ \rho^{2,2}=	-0.3270.$
}}
\label{Fig:BasketRandomCorrelation9M-Sept2015}
\end{figure}
\begin{table}[htb]
\begin{center}
\begin{tabular}{|c|c|c|}
\hline
\multicolumn {3}{|c|}{$T = 9$ Months}\\
\hline
K &  Shifted MVMD  & Shifted MVMDRC \\
\hline
112.84 & 0.0125 & 0.0069 \\ 
116.99 & 0.0011 &   0.0003  \\
121.04 & 0.0011&  0.0002 \\ 
126.18 & 0.0002 & $3.29*10^{-5}$ \\
141.66 & $3.46*10^{-5}$ & $2.2*10^{-5}$  \\
\hline 
\end{tabular}\\
\vspace{0.5cm}
\begin{tabular}{|c|c|c|}
\hline
\multicolumn {3}{|c|}{$T = 9$ Months}\\
\hline
K & Shifted MVMD  & Shifted MVMDRC \\
\hline
112.84 & 0.0054 & 0.0028 \\ 
116.99 & 0.0013 &   0.0004  \\
121.04 & 0.004 &  0.0007 \\ 
126.18 & 0.0026 & 0.0005 \\
141.66 & 0.01 & 0.0003  \\
\hline 
\end{tabular}
\end{center}
\caption{\textit{Calibration on $9$ months options, relative to 7th September 2015. The tables report absolute differences between market and model prices (top) and absolute differences between market and model implied volatilities (bottom).}}
\label{Table:BasketRandomCorrelation9M-Sept2015}
\end{table}

\section{Potential use for smile extrapolation: A Renminbi case study}

As a numerical illustration we performed on the shifted MVMD model analogous experiments as those in Section \ref{sec:Comparison} and Section \ref{Comparison-RandomCorrelation}, using 6 months options relative to data from 18th December 2015.  We first independently calibrate the exchange rates $S_1$=EUR/USD and $S_2$=USD/CNH on the corresponding implied volatilities, thus obtaining
$$\eta_1=(0.1132,0.0841),\ \lambda_1=(0.2209,0.7791),\ \beta_1=0.0063$$
$$\eta_2=(0.0447,0.1455),\ \lambda_2=(0.5956,0.4044),\ \beta_2=-0.0587$$
where the initial values of the single FX rates are $S_1(0)=1.0842$ and $S_2(0)=6.55$. The corresponding plots are shown in Figure \ref{Fig:ImpliedVol-RenminbiCase} (top).
 
In order to test how the model perform when looking at the FX cross rate $S_3=S_1 S_2=$EUR/CNH, we first calibrate the shifted  MVMD model using the atm call option only, thus obtaining 
$$\rho_{MVMD}(6M)=-0.1205.$$
We then plot the whole volatility curve corresponding to the calibrated value of $\rho$. This is represented in Figure \ref{Fig:ImpliedVol-RenminbiCase} (bottom, left) which shows a good fit.
Moreover, if we try to calibrate using not only the atm option, but a certain number of different strikes, we obtained values of the calibrated correlations quite close to the previous value $\rho_{MVMD}(6M)$. 

As a second experiment, we perform a calibration of the shifted MVMD model, when introducing random correlations. The values obtained are
$$\rho^{1,1}=0.864, \rho^{1,2}=-0.2843,\ \rho^{2,1}=-0.3417,\ \rho^{2,2}=-0.1006$$
and the corresponding plot is shown in Figure \ref{Fig:ImpliedVol-RenminbiCase} (bottom, right).
The values of expected random correlation and standard deviation under the $\bQ$ measure are 
$$\bE^{\bQ}( \rho^{i,j})=-0.1020$$
$$\mbox{Std}^{\bQ}(\rho^{i,j})=0.3904.$$

\begin{figure}[h!]
\begin{minipage}[b]{\linewidth}
\centering
\includegraphics[scale=0.34]{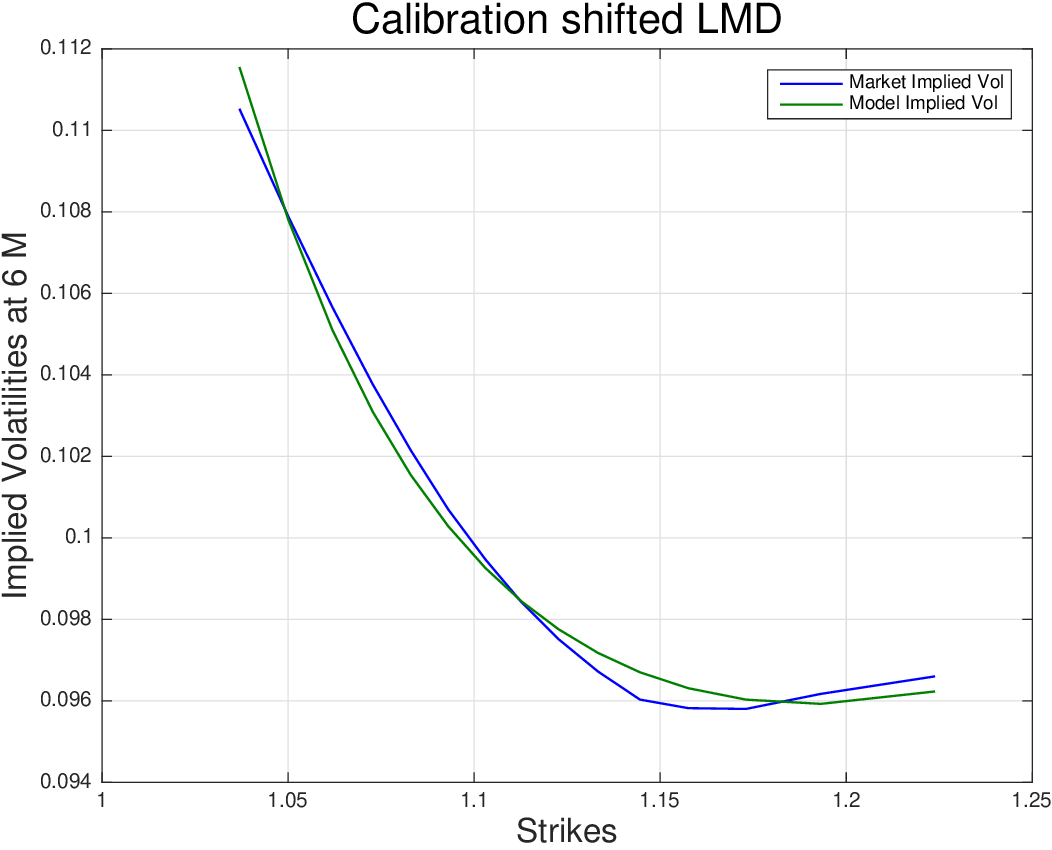}
\centering
\includegraphics[scale=0.34]{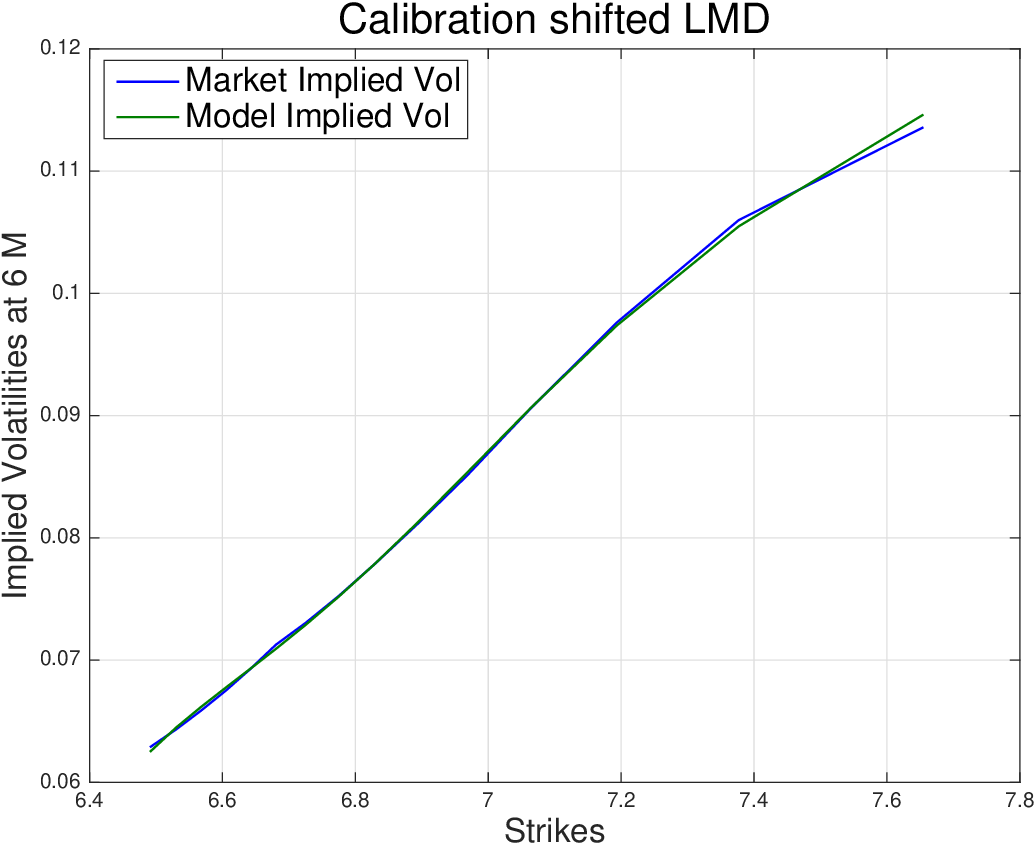}
\end{minipage}
\begin{minipage}[b]{\linewidth}
\centering
\includegraphics[scale=0.34]{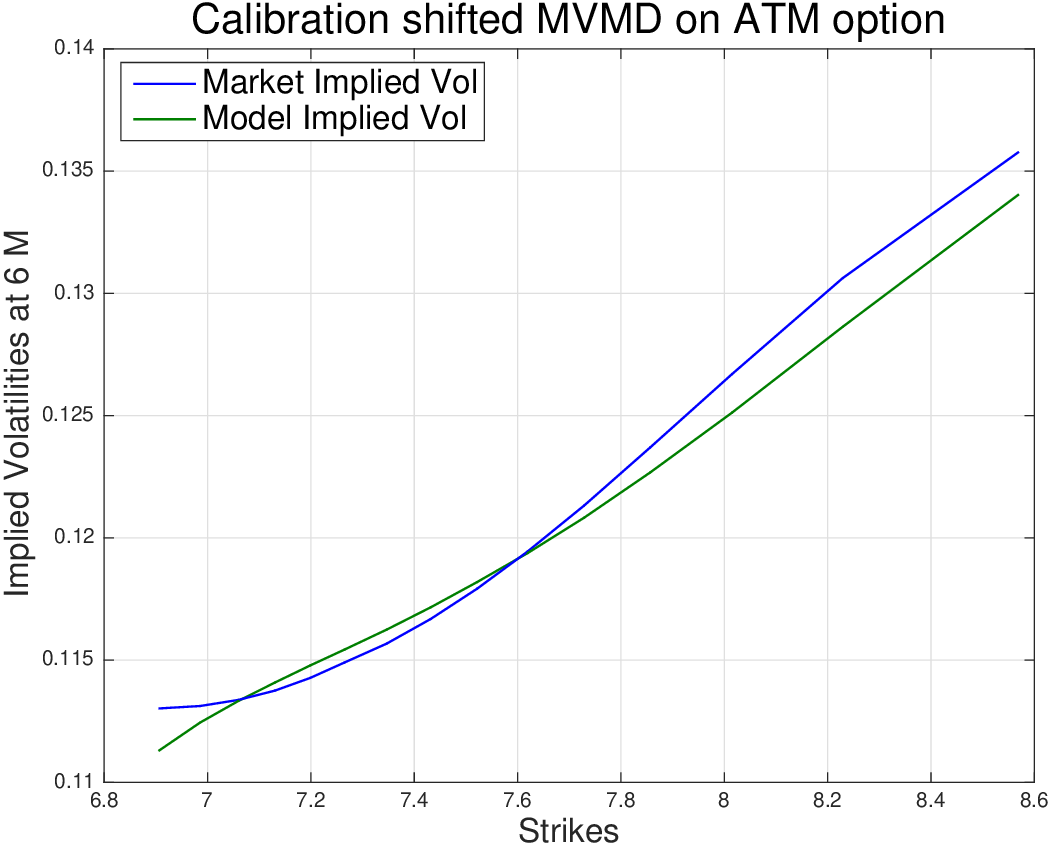}
\centering
\includegraphics[scale=0.34]{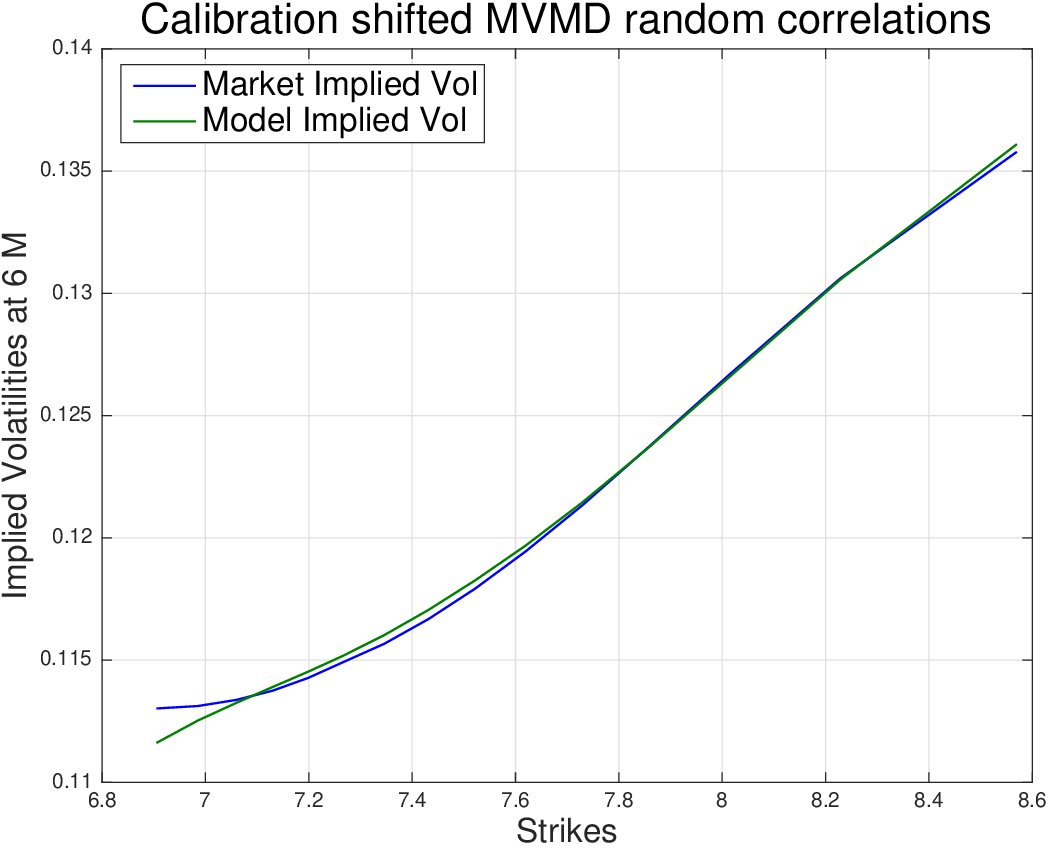}
\end{minipage}
\caption{\textit{On top, calibration of the EUR/USD FX rate (left) and the USD/CNH FX rate (right) each independently on a shifted LMD model. We used data from 18 December 2015. On the bottom, calibration of the MVMD model on the EUR/CNH exchange. In the case on the left, the calibration is obtained by fitting the shifted MVMD model to the atm option only, obtaining $\rho=-0.1205$. In the case on the right, the calibration is obtained by fitting the shifted MVMD model with random correlation. In the last case  $\rho_1=0.864$, $\rho_2=-0.2843$, $\rho_3=-0.3417$, $\rho_4=-0.1006$.}}
\label{Fig:ImpliedVol-RenminbiCase}
\end{figure}


\section{Conclusions}
We introduced a shifted MVMD model where each single asset follows shifted LMD dynamics which are combined so that the mixture property is lifted to a multivariate level, in the same way as for the non-shifted case \cite{MVMDWorkingPaper}. In this framework, we analysed the implied correlation from cross exchange rates and compared the results with those in the shifted SCMD model where the single assets are connected by simply introducing instantaneous correlations among the Brownian motions driving each asset.

 Finally, we generalized the MUVM model in \cite{MVMDWorkingPaper}, having MVMD as a Markovian projection, to a shifted model with random correlation, achieving more flexibility. This allows one to capture the correlation skew better. Indeed, the numerical experiments which we have conducted show that this model may be able to consistently reproduce triangular relationships among FX cross rates, or in other words to reproduce the implied volatility of a cross exchange rate in a consistent way with the implied volatilities of the single exchange rates.

One possible further use of the models given here is in proxying the smile for illiquid cross FX rates resulting from the product of two liquid FX rates. While one would have to find the relevant correlation parameters, possibly based on historical estimation with some adjustments for risk premia, the models presented here allow us to infer the detailed structure of the cross FX rate smile in an arbitrage free way.

\section{Appendix}
In this Appendix we provide the details leading to definition \ref{def:ShiftedMVMD}.
We start by applying a shift to each component $Y_i^k$ of each asset as follows $$S_i^k(t)=Y_i^k(t)+\beta_i e^{\mu_i t}.$$
Keeping in mind that $Y_i^k$ satisfies
\begin{equation}
dY_i^k(t)=\mu_i Y_i^k(t) dt+\sigma_i^k(t) Y_i^k(t) dZ_i(t)\,\ d\langle Z_i,Z_j\rangle=\rho_{ij}dt,
\end{equation}
we obtain, by applying Ito's formula
\begin{equation}
dS_i^k(t)= \mu_i S_i^k(t) dt + \sigma_i^k(t)\left(S_i^k(t)-\beta_i e^{\mu_i t}\right)d Z_i(t).
\end{equation}
The corresponding asset price $S_i$ will therefore be a shifted LMD model with shift equal to $\beta_i e^{\mu_i t}$. In order to find the dynamics of the whole multidimensional process $S(t)$, that is the process corresponding to $S(t)$ after having applied the shift, we look for an SDE of the type 
\begin{equation}
d \underline{S}(t)=diag(\underline{\mu})\underline{S}(t) dt + diag(\underline{S}(t)) \widetilde  C(t, \underline{S}(t))B d \underline{W}(t)
\end{equation}
where $\rho=B B^T$ such that the corresponding density satisfies
\begin{align}\label{eq:densityShift}
&p_{\underline{S}(t)}(\underline{x})=\sum_{k_1,k_2,\dots k_n = 1}^N \lambda_1^{k_1} \cdots \lambda_n^{k_n} \tilde{\ell}_{1, \dots n ;t}^{k_1, \dots , k_n} (\underline{x})\, \\
&\tilde{\ell}_{1, \dots n;t}^{k_1, \dots , k_n} (\underline{x}) = p_{[S_1^{k_1}(t), \dots , S_n^{k_n}(t)]^T}( \underline{x}).
\end{align}
In other words, the density $p_{S(t)}$ is obtained by mixing the single densities $p_{S_i^k(t)}(x)$ in all the possible ways.

 In order to find the diffusion matrix $\widetilde  C$, we compute the Fokker-Planck equations for $p_{\underline{S}(t)}$ and $\tilde{\ell}_{1, \dots n;t}^{k_1, \dots , k_n}$.
Defining $\tilde a(t,\underline{S}(t))=(\widetilde  C B) (\widetilde  C B)^T$ where $\widetilde  C_i$ denotes the i-th row of $\widetilde  C$, we obtain
\begin{equation}\label{eq:FokkerPlanckA(t)}
\frac{\partial}{\partial t}p_{\underline{S}(t)}(x)=- \sum_{i=1}^n \frac{\partial}{\partial x_i}\left[ \mu_i x_i p_{\underline{S}(t)}(x) \right] + \frac{1}{2} \sum_{i,j=1}^n \frac{\partial^2}{\partial x_i \partial x_j}\left[\tilde a_{ij}(t, \underline{x}) x_i x_j p_{\underline{S}(t)}(x) \right]
\end{equation}
and
\begin{multline*}
\frac{\partial \tilde{\ell}_{1, \dots n;t}^{k_1, \dots , k_n} (\underline{x})}{\partial t}= - \sum_{i=1}^n \frac{\partial}{\partial x_i}\left(\mu_i^{k_i}x_i \tilde{\ell}_{1, \dots n;t}^{k_1, \dots , k_n} (\underline{x}) \right)
\\ + \frac{1}{2} \sum_{i,j=1}^n \frac{\partial^2}{\partial x_i \partial x_j} \sigma_i^{k_i}(t)(x_i-\beta_i e^{\mu_i^{k_i}})\sigma_j^{k_j}(t)(x_j-\beta_j e^{\mu_j^{k_j} t}) \rho_{i,j} \tilde{\ell}_{1, \dots n;t}^{k_1, \dots , k_n} (\underline{x}).
\end{multline*}
Making use of equation (\ref{eq:densityShift}) and the equation above
\begin{multline*}
\frac{\partial }{\partial t}p_{\underline{S}(t)}(x)= \sum_{k_1,k_2,\dots k_n = 1}^N \lambda_1^{k_1} \cdots \lambda_n^{k_n} \frac{\partial }{\partial t} \tilde{\ell}_{1, \dots n ;t}^{k_1, \dots , k_n} (\underline{x})= \\
=  \sum_{k_1,k_2,\dots k_n = 1}^N \lambda_1^{k_1} \cdots \lambda_n^{k_n} \Bigl[ - \sum_{i=1}^n \frac{\partial}{\partial x_i}\left(\mu_i x_i \tilde{\ell}_{1, \dots n;t}^{k_1, \dots , k_n} (\underline{x}) \right)
\\ + \frac{1}{2} \sum_{i,j=1}^n \frac{\partial^2}{\partial x_i \partial x_j} \sigma_i^{k_i}(t)(x_i-\beta_i e^{\mu_i})\sigma_j^{k_j}(t)(x_j-\beta_j e^{\mu_j t}) \rho_{i,j} \tilde{\ell}_{1, \dots n;t}^{k_1, \dots , k_n} (\underline{x})\Bigr].
\end{multline*}
On the other hand, from equation (\ref{eq:FokkerPlanckA(t)})
\begin{multline*}
\frac{\partial }{\partial t}p_{\underline{S}(t)}(x)=- \sum_{i=1}^n \frac{\partial}{\partial x_i}\left[ \mu_i x_i \left(\sum_{k_1,k_2,\dots k_n = 1}^N \lambda_1^{k_1} \cdots \lambda_n^{k_n} \tilde{\ell}_{1, \dots n ;t}^{k_1, \dots , k_n} (\underline{x}) \right) \right] \\ + \frac{1}{2} \sum_{i,j=1}^n \frac{\partial^2}{\partial x_i \partial x_j}\left[\tilde a_{ij}(t, \underline{x}) x_i x_j \left( \sum_{k_1,k_2,\dots k_n = 1}^N \lambda_1^{k_1} \cdots \lambda_n^{k_n} \tilde{\ell}_{1, \dots n ;t}^{k_1, \dots , k_n} (\underline{x}) \right) \right].
\end{multline*}
Finally, comparing the two expressions obtained for $\frac{\partial }{\partial t}p_{\underline{S}(t)}(x)$
\begin{align*}
&\frac{1}{2} \sum_{i,j=1}^n \frac{\partial^2}{\partial x_i \partial x_j} \sum_{k_1,k_2,\dots k_n = 1}^N \lambda_1^{k_1} \cdots \lambda_n^{k_n}\Bigl[\tilde a_{ij}(t, \underline{x}) x_i x_j - \\ 
&\sigma_i^{k_i}(t)(x_i-\beta_i e^{\mu_i})\sigma_j^{k_j}(t)(x_j-\beta_j e^{\mu_j t})\rho_{i,j} \Bigr] \tilde{\ell}_{1, \dots n;t}^{k_1, \dots , k_n} (\underline{x})=0
\end{align*}

so that
\begin{equation*}
 a_{ij}=\frac{ \sum_{k_1,k_2,\dots k_n = 1}^N \lambda_1^{k_1} \cdots \lambda_n^{k_n}  \sigma_i^{k_i}(t)(x_i-\beta_i e^{\mu_i})\sigma_j^{k_j}(t)(x_j-\beta_j e^{\mu_j t}) \rho_{i,j} \tilde{\ell}_{1, \dots n;t}^{k_1, \dots , k_n} (\underline{x})}{x_i x_j \sum_{k_1,k_2,\dots k_n = 1}^N \lambda_1^{k_1} \cdots \lambda_n^{k_n}  \tilde{\ell}_{1, \dots n;t}^{k_1, \dots , k_n} (\underline{x})} .
\end{equation*}

\end{document}